\documentclass[acmsmall]{acmart}

\usepackage{booktabs} 
\usepackage{hyperref}   
\usepackage{url}  

\usepackage{multirow}  
\usepackage{color}  

\usepackage{diagbox}
\usepackage{amsmath}
\usepackage{arydshln} 
\usepackage{color} 

\usepackage[linesnumbered,ruled]{algorithm2e} 

\SetAlFnt{\small}
\SetAlCapFnt{\small}
\SetAlCapNameFnt{\small}
\SetAlCapHSkip{0pt}
\IncMargin{-\parindent}

\acmJournal{TKDD}

\setcopyright{acmcopyright}

\begin{document}
	
\title{HUSP-SP: Faster Utility Mining on Sequence Data}

\author{Chunkai Zhang}
\affiliation{
	\institution{Harbin Institute of Technology (Shenzhen)}
	\city{Shenzhen}
	\country{China}
}

\email{ckzhang@hit.edu.cn}

\author{Yuting Yang}
\affiliation{
	\institution{Harbin Institute of Technology (Shenzhen)}
	\city{Shenzhen}
 	\country{China}
}
\email{20S151137@stu.hit.edu.cn}

\author{Zilin Du}
 \affiliation{
 	\institution{Harbin Institute of Technology (Shenzhen)}
 	\city{Shenzhen}
 	\country{China}
 }
\email{yorickzilindu@gmail.com}

\author{Wensheng Gan}
\authornote{This is the corresponding author}
\affiliation{
	\institution{Jinan University}
	\city{Guangzhou}
	\country{China}
}
\email{wsgan001@gmail.com}

\author{Philip S. Yu}
\affiliation{
	\institution{University of Illinois at Chicago}
	\city{Chicago}
	\country{USA}
}
\email{psyu@uic.edu}

\begin{abstract}
	
High-utility sequential pattern mining (HUSPM) has emerged as an important topic due to its wide application and considerable popularity. However, due to the combinatorial explosion of the search space when the HUSPM problem encounters a low utility threshold or large-scale data, it may be time-consuming and memory-costly to address the HUSPM problem. Several algorithms have been proposed for addressing this problem, but they still cost a lot in terms of running time and memory usage. In this paper, to further solve this problem efficiently, we design a compact structure called sequence projection (seqPro) and propose an efficient algorithm, namely discovering high-utility sequential patterns with the seqPro structure (HUSP-SP). HUSP-SP utilizes the compact seq-array to store the necessary information in a sequence database. The seqPro structure is designed to efficiently calculate candidate patterns' utilities and upper bound values. Furthermore, a new upper bound on utility, namely tighter reduced sequence utility (TRSU) and two pruning strategies in search space, are utilized to improve the mining performance of HUSP-SP. Experimental results on both synthetic and real-life datasets show that HUSP-SP can significantly outperform the state-of-the-art algorithms in terms of running time, memory usage, search space pruning efficiency, and scalability.
	
\end{abstract}

%
%





\maketitle

\renewcommand{\shortauthors}{C. Zhang et al.}

\section{Introduction}

In the era of big data, sequence data is commonly seen in many domains. Sequential pattern mining (SPM) \cite{agrawal1995mining,han2001prefixspan, fournier2017surveys,van2018mining}, which extracts frequent subsequences from sequence database, is an interesting and important data mining topic in knowledge discovery in databases (KDD) \cite{chen1996data}. In the past two decades, SPM and the earlier developed frequent pattern mining (FPM) \cite{agrawal1994fast, han2004mining} have been applied in various domains, such as market basket analysis \cite{brin1997dynamic}, biology \cite{bejerano1999modeling}, weblog mining \cite{ahmed2010mining}, and natural language processing \cite{jindal2006identifying}. The main selection criteria for sequential patterns in SPM is co-occurrence (or the support measure): sequences with a high frequency are considered to be more interesting \cite{agrawal1995mining}. This characteristic contributes to the downward closure property (also known as the Apriori property \cite{agrawal1994fast}), which plays a fundamental role in the search space pruning of SPM algorithms and FPM algorithms, respectively. Note that SPM generalizes FPM by considering the sequential ordering of sequences. The frequency-based models (SPM and FPM) fit some fields, and field problems can generally be handled efficiently due to the inherent Apriori property. However, the effectiveness of data mining algorithms always takes precedence over efficiency. As the frequency measure only considers the objectivity of the data, the frequency-based models are no longer applicable when the subjective variables are also regarded as important factors. For instance, the patterns involving clingfilm are more likely to be obtained than those involving the refrigerator when the frequency-based model is applied to department store sales data. The managers may be more interested in the patterns containing refrigerators than clingfilm patterns, since they are more likely to think that refrigerators are much more profitable than clingfilm. To address this issue, another more comprehensive criteria, called utility, is introduced to utility-oriented pattern mining (or utility mining for short) \cite{gan2021survey}. For example, there have been many studies related to high-utility itemset mining (HUIM, from the perspective of transaction data) \cite{chan2003mining, lin2016efficient}, high-utility sequential pattern mining (HUSPM, from the perspective of sequence data) \cite{yin2012uspan, yin2013efficiently, wang2016efficiently,gan2018privacy}, and utility-oriented episode mining (HUEM, from the perspective of sequence events) \cite{gan2019utility}. In addition, of course, there are many case studies for incorporating utility mining into real-life situations.

In summary, HUSPM considers quantity, sequential order, and utility, which can usually provide the user with more informative and comprehensive patterns. However, there are more challenges in solving the HUSPM as described below. (1) The combinatorial explosion of sequences and utility computation in HUSPM. (2) HUSPM cannot utilize the downward closure property of Apriori \cite{agrawal1994fast} to efficiently prune the search space like the frequency-based algorithms. (3) The utility calculation task is far more difficult than frequency computation. More details about addressing the HUSPM problem have been summarized in \cite{zhang2021tkus,zhang2021shelf,gan2021survey}. Until now, there are several works have focused on this issue and developed efficient algorithms to extract the complete set of high-utility sequential patterns (HUSPs) from sequence data, such as USpan \cite{yin2012uspan}, HuspExt \cite{alkan2015crom}, HUS-Span \cite{wang2016efficiently}, ProUM \cite{gan2020proum}, and HUSP-ULL \cite{gan2020fast}. In addition, many interesting issues of effectiveness in utility-oriented SPM have been extensively studied, including the top-$k$ model \cite{wang2016efficiently,zhang2021tkus}, on-shelf availability \cite{zhang2021shelf}, explainable HUSPM \cite{gan2021explainable}, etc. HUSPM's existing work has developed a variety of data structures and pruning methods to improve mining efficiency. However, they are still expensive in terms of execution time and memory usage. In particular, in the era of big data, the characteristics of big data usually require data mining algorithms or models to have high efficiency and suitable effectiveness.

In order to improve the efficiency of utility-oriented sequence mining, we propose a faster algorithm, abbreviated as HUSP-SP, for discovering high-utility sequential patterns with a sequence projection (seqPro) structure and a new upper bound. The main contributions of this study can be summarized as follows:

\begin{itemize}
	\item 	\textbf{A compact data structure}.  The designed seqPro consists of a sequence-array (seq-array) and an extension-list, and it stores the corresponding projected database and some auxiliary information for each candidate pattern. Based on the seqPro structure, HUSP-SP can efficiently generate the candidate patterns and then calculate the candidate patterns' real utility and upper bound values.
	
	\item \textbf{A new utility upper bound}. We propose a tighter reduced sequence utility (TRSU) and provide detailed proof that TRSU is much tighter than all the existing upper bounds for HUSPM, such as reduced sequence utility (RSU) \cite{wang2016efficiently}.
	
	\item \textbf{Powerful pruning strategy}. A search space pruning strategy is designed, namely the early pruning (EP) strategy based on TRSU. Generally, the combination of the irrelevant items pruning (IIP) strategy \cite{gan2020fast} and the EP strategy can greatly reduce the search space.
	
	\item \textbf{High efficiency}.	Experimental results show that HUSP-SP can efficiently discover the complete set of HUSPs, and is superior to existing advanced algorithms in execution time, memory usage, search space pruning efficiency, and scalability.
\end{itemize}

The rest of this article is organized as follows. Section \ref{sec:relatedwork} briefly reviews the work of HUIM and HUSPM. Section \ref{sec:preliminaries} gives the basic definitions and formulates the HUSPM problem. In Section \ref{sec:method}, we propose a new data structure, upper bounds on sequence utility, and pruning strategies, and finally introduce all the details of the HUSP-SP algorithm. The experimental evaluation of the HUSP-SP method is given in Section \ref{sec:experiments}. Finally, the conclusion and future work are discussed in Section \ref{sec:conclusion}.
\section{Related Work}  \label{sec:relatedwork}

In this section, we separately review the literature on high-utility itemset mining and high-utility sequential pattern mining.

\subsection{High-Utility Itemset Mining}

In utility mining, each item/object in the database is associated with a utility to represent its importance. The task of HUIM \cite{chan2003mining} is to discover the complete set of high-utility itemsets (HUIs) consisting of items whose sum of utility values is higher than the predefined minimum utility threshold. Because HUIM no longer has the anti-monotone property of frequent pattern mining or association rule mining (ARM) \cite{agrawal1994fast}, mining HUIs has become extremely difficult as a result of the loss of a powerful search space pruning strategy. Liu \textit{et al.} \cite{liu2005two} handled this problem well with the transaction-weighted downward closure (TWDC) property. The Two-Phase algorithm \cite{liu2005two} efficiently extracts the candidate set of HUIs, called high transaction-weighted utilization itemsets (HTWUIs), based on the TWDC property in phase I. And then, for mining the real HUIs, only one extra database scan is performed to filter the overestimated itemsets. Inspired by the Two-Phase algorithm, some tree-based algorithms (e.g., IHUP \cite{ahmed2009efficient}, UP-Growth \cite{tseng2010up}, and UP-Growth$^+$ \cite{tseng2013efficient}) were developed to achieve better performance by reducing the number of HTWUIs that are generated in phase I. The main drawback of these algorithms is that they need to generate a large number of candidates in phase I, which results in poor runtime and memory performance. After that, several algorithms without two-phase operations, such as HUI-Miner \cite{liu2012mining}, FHM \cite{fournier2014fhm}, and EFIM \cite{zida2015efim}, were proposed. They can quickly discover HUIs without the candidate generation phase, which reduces the costly generation and utility computation of a large number of candidates. Besides the concentration on mining efficiency, there are also many works that concentrate on addressing the effectiveness issue of HUIM, such as mining the concise representations of HUIs to address the obscure mining result problem (large number while many of them are redundant) when the minimum utility threshold is set too low \cite{dam2019cls,shie2012efficient}; mining the top-$k$ HUIs without setting the minimum utility threshold \cite{tseng2015efficient}; exploiting the correlated HUIs that are not redundant \cite{gan2019correlated}, and so on. Several comprehensive surveys \cite{fournier2022pattern,gan2018surveyii,gan2021survey} provide more details on utility mining advancements.

\subsection{High-Utility Sequential Pattern Mining}

The utility-driven HUSPM \cite{gan2021survey,gan2018privacy} is an emerging problem that incorporates the utility concept into sequential pattern mining (SPM) \cite{van2018mining,fournier2017surveys,gan2019survey} to extract informative patterns by considering rich information, including sequential order and utility. In the past, HUSPM has been used for extracting web page traversal path patterns, web access sequences \cite{ahmed2010mining}, and high utility mobile sequences \cite{tseng2013efficient}. Similar to HUIM, the widely used downward closure property (also called the Apriori property \cite{agrawal1995mining}) in FPM and SPM is disabled by the introduction of utility, which makes the traditional pruning methods no longer applicable to HUSPM. Ahmed \textit{et al.} \cite{ahmed2010novel} first proposed two algorithms, UtilityLevel and UtilitySpan. As two-phase algorithms, both of them generate candidates by using a utility upper bound called sequence-weighted utilization (SWU) to prune the search space in phase I and then compute the exact utilities of each candidate. However, there was still no unified definition of HUSPM until Yin \textit{et al.} \cite{yin2012uspan} proposed a generic HUSPM framework and the USpan algorithm. In USpan, a utility-matrix structure was designed to compact the sequences, along with the sequence order and utility information of the processed database. The lexicographic quantitative sequence tree (LQS-tree) was introduced to represent the search space of HUSPM. Each node in the LQS-tree was denoted as a sequential pattern, and USpan performed the mining process to traverse the LQS-tree and extract the HUSPs. Additionally, two upper bounds, the SWU and sequence-projected utilization (SPU), were utilized for pruning the search space. However, the newly designed SPU-based pruning strategy has been proven to miss the real HUSPs in some conditions \cite{truong2019survey,gan2020proum} which means USpan and other SPU-based algorithms can not return complete HUSPs. Some works with tighter upper bounds and more compact data structures were proposed to improve the mining efficiency. The HUS-Span algorithm \cite{wang2016efficiently} adopts two tighter upper bounds: the prefix extension utility (PEU) for depth pruning and the reduced sequence utility (RSU) for breadth pruning. The ProUM algorithm \cite{gan2020proum} comes up with a new upper bound, the sequence extension utility (SEU), and a compact utility-array structure.

Besides, to simplify the parameter settings, e.g., the minimum utility threshold, some studies focus on mining top-$k$ HUSPs \cite{yin2013efficiently,wang2016efficiently,zhang2021tkus}. Other interesting topics for utility mining have been extensively studied, such as incremental HUSPM \cite{wang2018incremental}, HUSPM over data streams \cite{zihayat2017memory}, HUSPM with individualized thresholds \cite{gan2020utility}, HUSPM with negative item values \cite{xu2017mining}, FMaxCloHUSM \cite{truong2019fmaxclohusm} for mining frequent closed and maximal high utility sequences, HUSPM on the Internet of Things \cite{srivastava2020large}, targeted high-utility sequence querying \cite{zhang2021tusq}, and self-adaptive high-average utility one-off SPM \cite{wu2021haop}. Gan \textit{et al.} \cite{gan2021utility} proposed two algorithms, MDUS$_{EM}$ and MDUS$_{SD}$, to discover multidimensional high-utility sequential patterns. On-shelf utility mining based on sequence data was also studied to improve on-shelf availability \cite{zhang2021shelf}. However, the above algorithms are still not efficient enough, especially when dealing with complex sequence data and low utility thresholds. Gan \textit{et al.} \cite{gan2020fast} integrated two pruning strategies, Look Ahead Removing (LAR) and Irrelevant Item Pruning (IIP) strategies, with the developed HUSP-ULL algorithm, which significantly improved the mining efficiency of the HUSPM problem when processing sequence data.

\section{Definitions and Problem Statement}
 \label{sec:preliminaries}
 
In this section, we detail the important definitions and notations used in this paper. Then, the problem statement of HUSPM is presented.

Let $I$ = $\{i_{1}$, $i_{2}$, $\ldots$, $i_{n}\}$ be a set of all distinct items in the database. An itemset $e$ is a subset of $I$. A sequence $s$ is an ordered list of itemsets (also called elements). Without loss of generality, the items within each element are sorted alphabetically, and the "$\prec$" is used to represent that one item occurs before another item in an element. Furthermore, the number of elements in a sequence $s$ is its size; the total number of items in $s$ is its length; and a sequence with a length of $k$ is known as a $k$-sequence. For example, a sequence $s$ = $\langle\{a\ b\}$, $\{b\ c\ d\}$, $\{a\ e\}\rangle$ composed of seven items and three elements. Then, the size of $s$ is three, and it's called a $7$-sequence for its length is seven. The sequence $t$ = $\langle$ $e_1$, $e_2$, $\ldots$, $e_p$$\rangle$ is a subsequence of $s$ = $\langle$$e_1'$, $e_2'$, $\ldots$, $e_q'$$\rangle$, or $s$ contains $t$, represented as $t \sqsubseteq s$, if and only if there exists integers $1 \le j_1$ $<$ $j_2$ $<$  $\ldots $ $<j_p$ $\le q$ such that $e_1$ $ \subseteq $ $e_{j_1}'$, $e_2$ $\subseteq$  $e_{j_2}'$, $\ldots$, $e_p$ $ \subseteq$  $e_{j_p}'$.

\begin{table}[h]
	\centering
	\caption{A running example of a quantitative sequential database, $\mathcal{D}$}
	\label{table1}
	\includegraphics[width=0.6\linewidth]{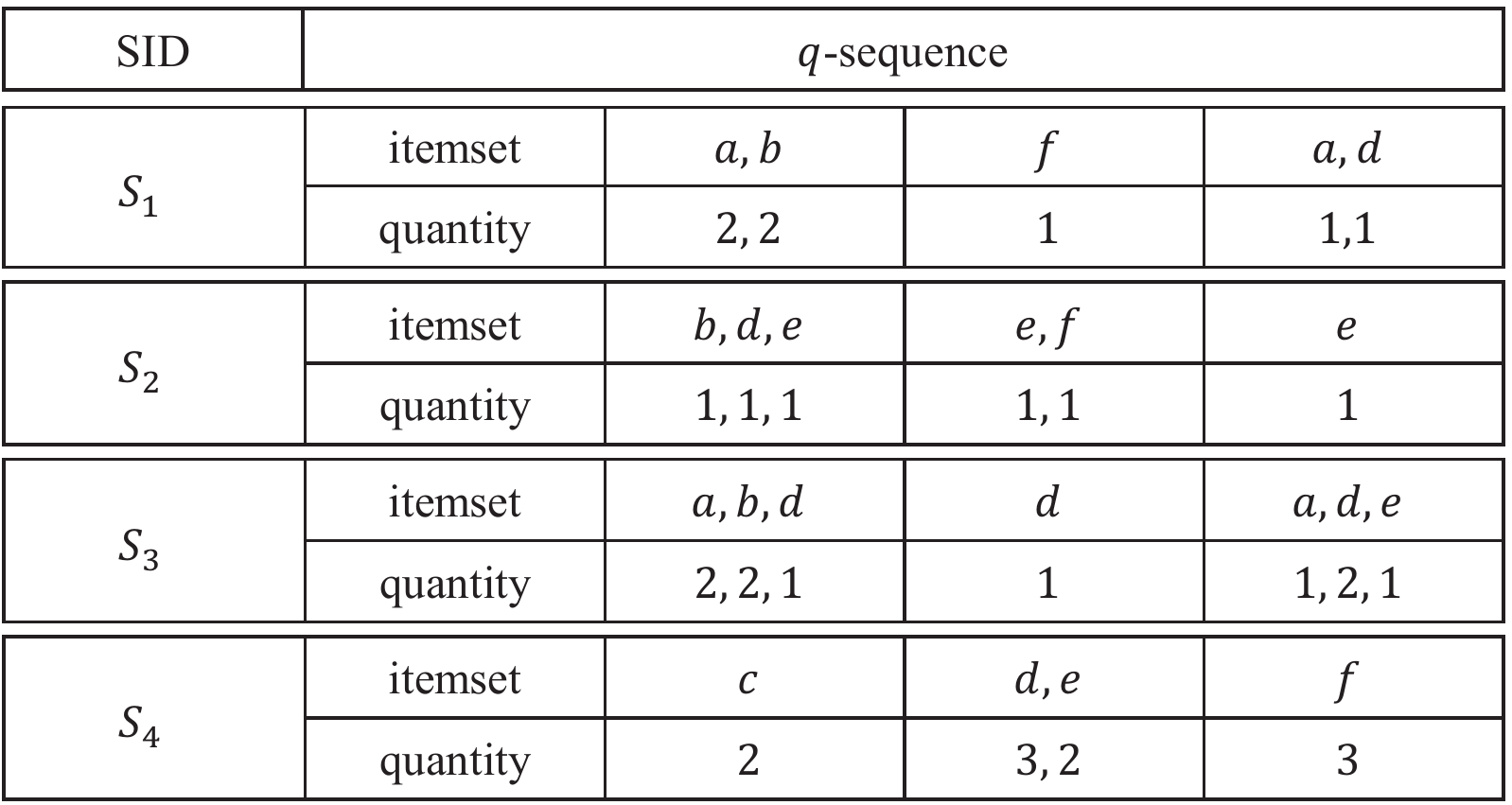}
\end{table}

\begin{definition}(quantitative sequential database and quantitative sequence)
	 A quantitative sequential database (QSDB) $\mathcal{D}$ is a set of tuples $(SID, S)$, where each sequence $S$ is associated with a unique identifier (SID). Moreover, the QSDB sequences express more information than the non-QSDB sequences. In a QSDB, each item $i$ within an element $e_j$ in a sequence $S$ = $\langle$$E_1$, $E_2$, $\ldots$, $E_p$$\rangle$ is associated with a positive integer $q(i,j,S)$, called the internal utility or quantity. Therefore, we give the sequence in a QSDB a new name, called the quantitative sequence or $q$-sequence to distinguish it from the non-QSDB whose items are not associated with quantity. Similarly, the items and itemsets are called $q$-item and $q$-itemset, respectively, in a QSDB. Each distinct item $i$ in a QSDB is associated with a unity utility (also known as external utility), which is denoted as $eu(i)$.
\end{definition}
	 
Table \ref{table1} shows a QSDB $\mathcal{D}$ with four $q$-sequence, and the unit utility (also called external utility) of each distinct item in $\mathcal{D}$ is \{$a$: 3, $b$: 1, $c$: 2, $d$: 1, $e$: 1, $f$: 1\}. They will be used as a running example in this paper.

\begin{definition}(utility of item, itemset and $q$-sequence in a QSDB)
    Given a QSDB $\mathcal{D}$, $q$-sequence $S$ = $\langle$$E_1$, $E_2$, $\ldots$, $E_n\rangle$, the utility of an item $i$ within the $j$th itemset $E_j$ is defined as $u(i,\ j,\ S)$ = $eu(i)$ $\times$ $q(i,\ j,\ S)$, where $1 \leq j$ $\leq n$. The utility of itemset $E_j$ and its subset $e_j$ $\subseteq $ $E_j$ is defined as $u(e_j,\ j,\ S)$ = $\sum_{i\in E_j\cap e_j}$$u(i,\ j,\ S)$, and the utility of the $q$-sequence $S$ is defined as $u(S)$ = $\sum_{i=1}^{n} u(E_i,\ i,\ S)$. Consequently, the utility of the QSDB $\mathcal{D}$ is $u(\mathcal{D})$ = $\sum_{\forall S\in \mathcal{D}}u(S)$.
\end{definition}

For example, in Table \ref{table1}, the utility of the item $a$ within the first $q$-itemset $E$ = $\{a\ b\}$ in $S_1$ is $u(a,1,S_1 )$ = $p(a)$ $\times$ $q(a,1,S_1 )$ = 3 $\times$ 2 = 6. Then the utility of $E$, $u(\{a\ b\},1,S _1)$ = 3 $\times$ 2 + 1 $\times$ 2 = 8, and the utility of $S_1 $ is $u(S_1 )$ = $u(\{a\ b\},1,S_1 )$ + $u(\{f\},S_1 )$ + $u(\{a\ d\},S_1 )$ = 8 + 1 + 4 = 13. Accordingly, the utility of $\mathcal{D}$, $u(\mathcal{D})$ = $u(S_1)$ + $u(S_2)$ + $u(S_3)$ + $u(S_4)$ = 13 + 6 + 16 + 12 = 47.

\begin{definition}(sequence instance)
	Given a sequence $t$ = $\langle$$e_1$, $e_2$, $\ldots$, $e_p\rangle$, $q$-sequence $S$ = $\langle E_1$, $E_2$, $\ldots$, $E_q\rangle$, $t \sqsubseteq S$ and $\exists$ $1 \le j_1$ $< j_2 < $ $\ldots$ $<j_p \le q$ that $e_1 \subseteq E_{j_1}$, $e_2$ $\subseteq$ $E_{j_2}$, $\ldots$, $e_p$ $\subseteq $ $E_{j_p}$, then it is said that the $q$-sequence $S$ has an instance of sequence $t$ at position $\langle$$j_1$, $j_2$, $\ldots$, $j_p$$\rangle$.
\end{definition}

For example, in Table \ref{table1}, for sequence $\langle\{a\},\{a\} \rangle$ and $S_1$  = $\langle\{a\ b\}$, $\{f\}$, $\{a\ d\}$$\rangle$, $\{a\}$ $\subseteq $ $\{a\ b\}$ and $\{a\} \subseteq \{a\ d\}$, thus $\langle\{a\},\{a\}$ $\rangle$ $\sqsubseteq$ $S_1 $ and $S_1 $ has an instance of sequence $\langle\{a\},\{a\} \rangle$ at $\langle 1,3 \rangle$.
	
\begin{definition}(instance utility)
	Assume $S$ has an instance of $t$ at position $\langle$$j_1$, $j_2$, $\ldots$, $j_p\rangle$, the utility of the instance of $t$ at position $\langle$$j_1$, $j_2$, $\ldots$, $j_p\rangle$ is defined as 
	\begin{equation}
		u(t,\ \langle j_1,j_2,\ldots,\ j_p\rangle ,\ S) = \sum_{i=1}^{p} u(e_i,j_i,S),
	\end{equation}	
	where $e_i$ $\subseteq$ $E_{j_i}$. 
\end{definition}

For example, in Table \ref{table1}, $u(\langle \{a\},\{a\}\rangle,\langle 1,3\rangle,S _1)$ = 6 + 3 = 9.

\begin{definition}(sequence utility) \label{sequtility}
	The utility of sequence $t$ in $q$-sequence $S$ is defined and denoted as $u(t, S)$ = $\max\{u(t,\ \langle j_1,j_2,\ldots,j_p\rangle,\  S)$ $|$ $\forall \langle j_1,j_2,\ldots,j_p$$\rangle:t \sqsubseteq \langle$ $E_{j_1}$, $E_{j_2}$, $\ldots$, $E_{j_p}\rangle \}$, that is the maximum utilities of all the instances of $t$ in $S$. Moreover, the utility of $t$ in a QSDB $\mathcal{D}$ is defined as 
	 \begin{equation}
	 	u(t, \mathcal{D}) = \sum_{\forall S \in \mathcal{D} \wedge t \sqsubseteq S} u(t,S).
	 \end{equation}	
\end{definition}

For example, in Table \ref{table1}, the utility of $\langle \{a\ d\}\rangle$ in $S _3$ is calculated as $u(\langle \{a\ d\}\rangle, S _3)$ = $\max \{u(\langle \{a\ d\}\rangle$, $\langle 1\rangle$, $S _3)$, $u(\langle \{a\ d\} \rangle$, $\langle 3\rangle$, $S _3)\}$ = $\max \{7, 5\}$ = 7. Accordingly, the utility of $\langle \{a\ d\}\rangle$ in $\mathcal{D}$ is $u(\langle \{a\ d\}\rangle, \mathcal{D})$ = $u(\langle \{a\ d\}\rangle, S _1)$ + $u(\langle \{a\ d\}\rangle, S _3)$ = 4 +  7 = 11.

\textbf{Problem statement}: Given a QSDB $\mathcal{D}$ and a user-specified minimum utility threshold $\xi$ from 0 to 1, a sequence $t$ with $u(t, \mathcal{D})$ $\ge$ $\xi$ $\times$ $u(\mathcal{D})$ is called a high-utility sequential pattern (HUSP). The high-utility sequential pattern mining (HUSPM) problem is to discover all the HUSPs in $\mathcal{D}$ with respect to the $\xi$.

For example, given the QSDB Table \ref{table1} and $\xi$ = 0.2, then the minimum utility $\xi \times$ $u(\mathcal{D})$ = 0.2 $\times$ 47 = 9.4. $\langle \{a\ d\}\rangle$ is a HUSP in $\mathcal{D}$, for $u(\langle \{a\ d\}\rangle, \mathcal{D})$ = 11 $\ge$ $9.4$, and the HUSPM problem is to find all the sequences contained in the $q$-sequences of $\mathcal{D}$ with utilities no less than 9.4.

\section{Proposed Algorithm}    \label{sec:method}

This section describes the proposed algorithm, HUSP-SP, for addressing the HUSPM problem by generating and testing the database's promising subsequences (patterns). Utilizing the newly proposed upper bound TRSU and the TRSU based EP pruning strategy, HUSP-SP only needs to generate and test a few patterns. Besides, the newly designed running data structure, called seqPro, significantly reduces the memory usage of an algorithm and facilitates the utility computation process. In general, HUSP-SP compacts the utility and sequence information of the QSDB into memory by using the seq-array structure called seqPro. Then HUSP-SP first finds the 1-sequences and then recursively projects the seqPro by prefix to find the more extended patterns. The entire searching process forms a lexicographic $q$-sequence (LQS)-tree.

\subsection{Search Space}

To avoid generating patterns that do not appear in the database and repeatedly testing the same patterns, HUSP-SP adopts the pattern-growth and projection database methods \cite{han2001prefixspan}, which means it starts from patterns composed of a single item and finds larger patterns by recursively appending items to discovered patterns. The LQS-tree \cite{yin2012uspan, wang2016efficiently} is used to formally describe the search space of HUSP-SP. An LQS-tree is shown in Fig. \ref{LQS-tree}, where each node represents a pattern and the lexicographically ordered child nodes are generated by first applying I-Extension and then S-Extension with the corresponding available 1-sequences, respectively.

\begin{definition}(I-Extension and S-Extension \cite{han2001prefixspan, yin2012uspan})
	Let $t$ = $\langle$$E_1$, $E_2$, $\ldots$, $E_p$$\rangle$, extension is the operation appending a sequence $w$ = $\langle$$F_1$, $F_2$, $\ldots$, $F_q\rangle$ to the end of $t$. Given $\forall i \in E_p$, $\forall i' \in F_1$, $i \prec i'$, I-Extension is defined as $t \diamondsuit_i w$ = $\langle$$E_1$, $E_2$, $\ldots$, $E_p$ $\cup$ $F_1$, $F_2$, $\ldots$, $F_q\rangle$. The S-Extension is defined as $t\diamondsuit_s w$ = $\langle$$E_1$, $E_2$, $\ldots$, $E_p$, $F_1$, $F_2$, $\ldots$, $F_q\rangle$. Additionally, notation $\diamondsuit$ can represent either I-Extension or S-Extension. 	
\end{definition}

For example, $\langle \{a\}, \{a\ b\}\rangle$ $\diamondsuit_i$ $\langle \{e\},\{a\ c\}\rangle$ = $\langle \{a\},\{a\ b\ e\},\{a\ c\} \rangle$; $\langle \{a\}, \{a\ b\}\rangle$ $\diamondsuit_s$ $\langle \{e\},\{a\ c\}\rangle$ = $\langle \{a\},\{a\ b\},\{e\}$, $\{a\ $ $c\} \rangle$, and $\langle \{a\}, \{a\ b\}\rangle$ $\diamondsuit $ $\langle \{e\},\{a\ c\}\rangle$ equals $\langle \{a\}, \{a\ b\}\rangle$ $ \diamondsuit_i $ $\langle \{e\},\{a\ c\}\rangle$ or $\langle \{a\}, \{a\ b\}\rangle$ $\diamondsuit_s $ $\langle \{e\},\{a\ c\}\rangle$. Note that $\langle \{a\}$, $\{a\ b\}\rangle $ $ \diamondsuit_i$ $\langle\{b\ d\}\rangle$ is forbidden for $b \prec b$ is not hold. 
 
Similar to the previous HUSPM algorithms \cite{yin2012uspan,yin2013efficiently,gan2020proum}, HUSP-SP traverses the LQS-tree in a depth-first manner and calculates the utility of the patterns with respect to the corresponding tree nodes. As shown in Fig. \ref{LQS-tree}, HUSP-SP starts at the empty root node and first finds the 1-sequences, which make up the first layer of the LQS-tree. Then HUSP-SP turns to the $\langle\{a\}\rangle$ node, checks whether $\langle \{a\} \rangle$ is a HUSP by calculating the utility of $\langle \{a\} \rangle$, and generates $\langle \{a\} \rangle$'s possible children. The same operation will be applied to the first children $\langle \{a\ b\} \rangle$. The generation and check processes will be recursively invoked until there is no other node that should be visited.

Consequently, for HUSPM methods to be effective and efficient, they usually need to be able to handle the following two problems well:

\begin{itemize}
    \item How to find the candidate items for the current testing pattern to generate an extended pattern and calculate the utility of extended patterns efficiently?
    
    \item How to reduce the search space?
\end{itemize}

\begin{figure*}[h]
    \centering
    \includegraphics[width=0.8\linewidth]{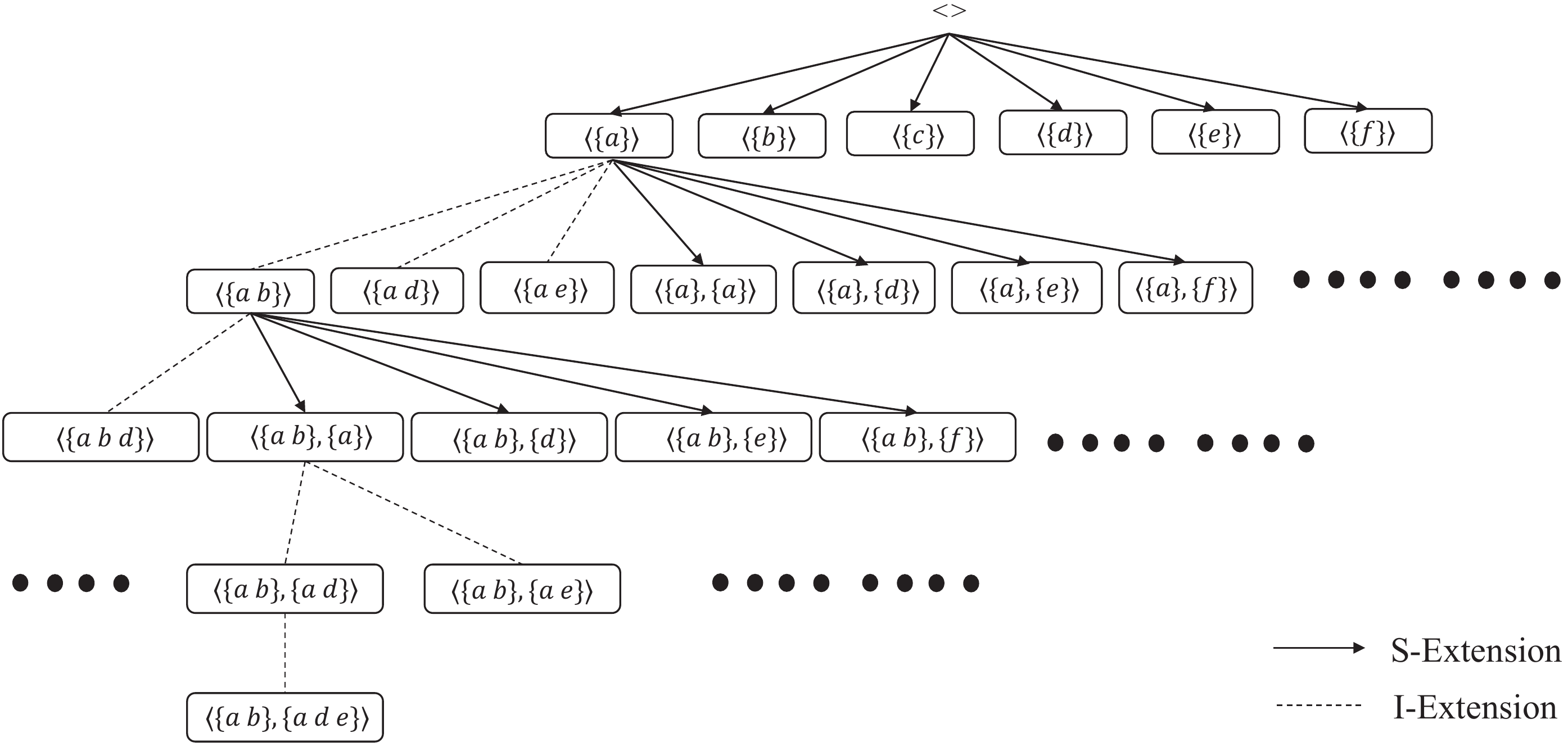}
    \caption{Partial of the LQS-tree for HUSPM with $\mathcal{D}$ in Table \ref{table1}}
    \label{LQS-tree}
\end{figure*}

\subsection{Pattern Generation and Utility Calculation}

Firstly, some necessary definitions and basic data structures for facilitating the description of the detailed methods are presented. Given a $k$-length $q$-sequence $S$, the items within $S$ are indexed from $1$ to $k$ by their sequential orders. For example, in Table \ref{table1}, $S_1$ is indexed from 1 to 5, and the item name and quantity of the $q$-item with index 3 are $f$ and $1$, respectively.

\begin{definition}(extension item and extension position \cite{wang2016efficiently})
	Given sequence $t$ = $\langle$$E_1$, $E_2$, $\ldots$, $E_p\rangle$, $q$-sequence $S$ = $\langle$$E_1'$, $E_2'$, $\ldots$, $E_q'$$\rangle$, and $S$ has an instance of $t$ at $\langle$$j_1$, $j_2$, $\ldots$, $j_p$$\rangle$. Then, the last item of $E_p$ is called the extension item, and $j_p$ is called an extension position of $t$ in $S$. In addition, the index of the extension item is denoted as $I(t,j_p)$.
\end{definition}

For example, in Table \ref{table1}, $S_3$ has three instances of sequence $\langle\{d\},\{d\} \rangle$ at $\langle1,2 \rangle$, $\langle1,3 \rangle$ and $\langle2,3 \rangle$ with common extension item $d$, and the corresponding extension positions are 2, 3, and 3, separately. Besides, the index of the extension item $d$ within extension position 3 is $I(\langle\{d\},\{d\} \rangle,3)$, and it equals 6 since it's the $6th$ $q$-item in $S_3$.

\begin{definition}(sequence utility with extension position \cite{wang2016efficiently}) \label{exutility}
	The utility of $t$ with extension position $p$ in $S$ is defined as the maximum utility of any instance of $t$ whose extension position is $p$. It is denoted as
	\begin{equation}
		u(t,\ p,\ S) = \max \{u(t,\ \langle j_1,j_2,\ldots,p \rangle,\ S) | \forall \langle j_1,j_2,\ldots,p\rangle:t \sqsubseteq \langle E'_{j_1},E'_{j_2},\ldots,E'_{_p}\rangle\}.
	\end{equation}
\end{definition}

\begin{definition}(remaining $q$-sequence \cite{yin2012uspan, wang2016efficiently})
	Given a $p$-length $q$-sequence $S$ and an item index $q$, $q$ $<$ $p$, the subsequence from the item at index $q+1$ to the end of $S$ denoted as $S/$$_q$, is called the remaining sequence of $S$ with respect to item index $q$.
\end{definition}

For example, in Table \ref{table1}, $u(\langle\{d\},\{d\} \rangle, 3, S_3  )$ = $\max \{u(\langle\{d\},\{d\} \rangle, \langle1,3 \rangle ,S_3 )$, $u(\langle\{d\},\{d\} \rangle, \langle2,3 \rangle ,S_3)\} $ = $\max \{3,3\}$ = 3. $S_3$ has an instance of $\langle \{a\ b\} \rangle$ at $\langle 1 \rangle$, then $S_3$/$_{I(\langle \{a\ b\}\rangle,1)}$ = $S_3/_2$ = $\langle$ $\{d\}$, $\{d\}$, $\{a\ d\ e\}$$\rangle $.

\begin{definition}(Sequence-array) \label{seqA}
     Given a $q$-sequence $S$ = $\langle$$E_1$, $E_2$, $\ldots$, $E_n\rangle$ with length $k$, the sequence-array (seq-array) of $S$ has four length $k$ arrays for storing the information of each item (\textit{item array} for item name, \textit{utility array} for item utility, \textit{remaining-utility array} for the utility of the remaining $q$-sequence with respect to current item index, and \textit{element-index array} for saving the index of the current element index, i.e., the index of first item of current element). In addition, the \textit{item-indices table} field of the seq-array records the indices of each distinct item within $S$. For instance, the seq-array of $S_1$ in Table \ref{table1} is shown in Fig. \ref{seq-array}.
\end{definition}

\begin{figure*}[h]
    \centering
    \includegraphics[width=0.7\linewidth]{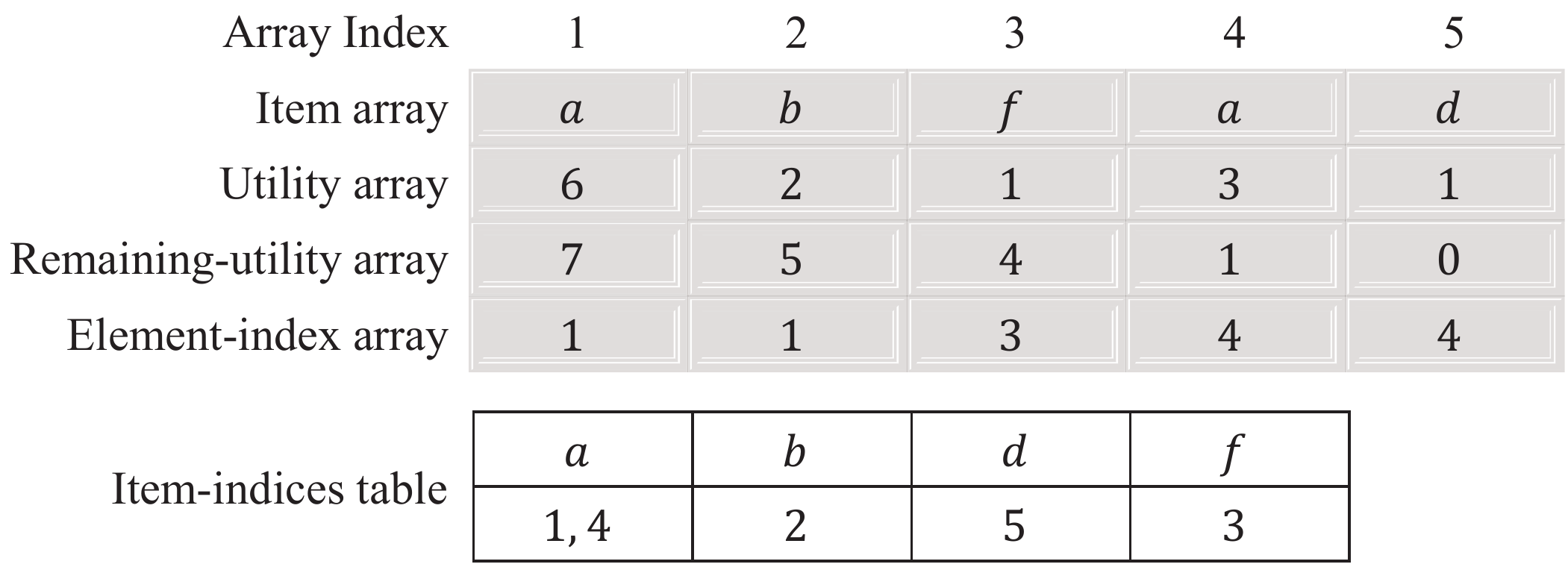}
    \caption{The seq-array of $S_1$ in Table \ref{table1}}
    \label{seq-array}
\end{figure*}

\begin{definition}(Extension-list)
	Assuming $q$-sequence $S$ of QSDB $\mathcal{D}$ has $k$ instances of sequence $t$ with extension positions $p_1$, $p_2$, $\ldots$, $p_k$, where $p_1$ $<$ $p_2$ $<$ $\ldots$ $<$ $p_k$. The extension-list of $t$ in $S$ consists of $k$ elements, where the $i$th element contains the following fields:
	\begin{itemize}
		\item Field \emph{acu} is the utility of $t$ with extension position $p_i$. 
		\item Field \emph{exIndex} is the $I(t,p_i)$ (the index of the extension item with extension position $p_i$).
	\end{itemize}
\end{definition}

For example, in Table \ref{table1}, $S_3$ has three instances of sequence $\langle\{d\},\{d\} \rangle$ at $\langle1,2 \rangle$, $\langle1,3 \rangle$ and $\langle2,3 \rangle$. Thus, the extension-list of $\langle\{d\},\{d\} \rangle$ in $S_3$ has two elements ``\textit{acu}: $u(\langle\{d\},\{d\} \rangle, 2, S_3)$, \textit{exIndex}: 4", and ``\textit{acu:} $u(\langle\{d\},\{d\} \rangle, 3, S_3)$, \textit{exIndex}: 6". For the projection database of each LQS-tree node, HUSP-SP introduces the Sequence Projection (seqPro) structure to represent the $q$-sequence. The seqPro structure is composed of two fields:
\begin{itemize}
    \item Seq-array: This is a pointer to the original seq-array of the $q$-sequence.
    
    \item Extension-list: It is a list that records the indices of extension items and the sequence utilities with extension positions of the current node pattern.
\end{itemize}

Compared to the projection database structure, such as UL-lists proposed in HUSP-ULL \cite{gan2020fast}, the seqPro is more refined in modeling the $q$-sequence. Furthermore, with the newly introduced item-index head table, seqPro avoids the problem of storing lots of null values in UL-lists. Additionally, the introduced Extension-list makes seqPro more useful in facilitating pattern generation and utility calculation.

According to Definition \ref{sequtility}, the utility of a pattern in a $q$-sequence is the maximum utility of all the pattern instances in the $q$-sequence, and the pattern utility is the sum of all the pattern utilities of $q$-sequences in the QSDB that contain the pattern. Obviously, it is time-consuming for the utility calculation for each candidate pattern to find all the instances by scanning the sequences, then compute the utility of each instance, select the maximum utilities, and sum them up at last. Thanks to pattern growth and the database projection methods, the utility computation complexity of child nodes is greatly reduced by recording the location and utility of the parent node pattern's instances (extension-list). For example, in Table \ref{table1}, the utility calculation process of sequence $\langle\{d\}, \{a\}\rangle$ can be described as follows. First, although the prefix sequence $\langle\{d\}\rangle$ is contained in every $q$-sequence in the database, only $S_3$ needs to be considered since only $S_3$ contains $\langle\{d\}, \{a\}\rangle$. Next, the utility of $\langle\{d\},\{a\}\rangle$ can be efficiently calculated in projected database $S_3$ based on the recorded instances of $\langle\{d\}\rangle$. Specifically, the sequence $\langle\{d\}, \{a\}\rangle$ can be formed by appending item $a$ to instances of $\langle\{d\}\rangle$ with extension positions 1 and 2, respectively. Notice that as there is only one extension item of $\langle\{d\}, \{a\}\rangle$ in $S_3$, $u(\langle\{d\}, \{a\}\rangle)$ = $u(\langle\{d\}, \{a\}\rangle, 3, S_3)$ = $\max\{u(\langle\{d\}\rangle, 1, S_3)$, $u(\langle\{d\}\rangle, 2, S_3)\}$ + $u(a, 3, S_3)$ = $\max\{1, 1\}$ + 3 = 4. Furthermore, the extension-list of $\langle\{d\},\{a\}\rangle$ is constructed for facilitating the candidate items discovery, utility calculation, and upper bound calculation of newly generated super-sequences.

\subsection{Search Space Pruning} \label{prune}

To reduce the search space, HUSP-SP not only adopts the utility upper bounds (prefix extension utility (PEU) \cite{wang2016efficiently}, and reduced sequence utility (RSU) \cite{wang2016efficiently}) but also proposes a new utility upper bound, called tighter reduced sequence utility (TRSU). Furthermore, based on these upper bounds, several powerful pruning strategies are designed in HUSP-SP.

\begin{definition}(PEU \cite{wang2016efficiently}) \label{peu}
	The \textit{PEU} of sequence $t$ in $q$-sequence $S$ with extension position $p$ is defined as 
	\begin{equation}
	\textit{PEU}(t,\ p,\ S) = \begin{cases}
		u(t,\ p,\ S) + u(S/_{I(t,p)}) , & \text{if } u(S/_{I(t,p)}) \ge 0; \\
		0, & \text{otherwise}.
	\end{cases}
	\end{equation}
	The \textit{PEU} of sequence $t$ in $S$ is defined as 
	\begin{equation}\label{eq3}
		\textit{PEU}(t,S) = \max \{\textit{PEU}(t,\ p,\ S) | \forall \text{extension position } p \text{ of } t \text{ in } S\}
	\end{equation}
	Note, \textit{PEU}($t,S$) = $u(S)$ when $t$ is null. Moreover, the \textit{PEU} of sequence $t$ in QSDB $\mathcal{D}$ is defined as 
	\begin{equation}
	\textit{PEU}(t, \mathcal{D}) = \sum_{\forall S \in \mathcal{D} \land t \sqsubseteq S} \textit{PEU}(t,S)
	\end{equation}
\end{definition}

For example, int Table \ref{table1}, $\textit{PEU}(\langle \{a\ b\} \rangle, \mathcal{D})$ = $u(\langle \{a\ b\} \rangle, 1, S  _1)$ + $u(S  _1 /_{I(\langle\{a\ b\} \rangle, 1 )})$ + $u(\langle \{a\ b\} \rangle, 1, S  _3)$ + $u(S_3$$/_{I(\langle\{a\ b\} \rangle, 1)})$ = 8 + 5 + 8 + 8 = 29.

\begin{theorem}\label{PEU_theorem}
	For any extension of sequence $t$ in database $\mathcal{D}$ with a sequence $w$ (not null), $t'$ = $t \diamondsuit w$, $u(t', \mathcal{D})$ $\le$ \textit{PEU}($t, \mathcal{D}$) \cite{wang2016efficiently}. 
\end{theorem}

\begin{proof}
	Let $\mathcal{S}_t \subseteq \mathcal{D}$, $\mathcal{S}_{t'}$ $\subseteq \mathcal{D}$ be the set of $q$-sequences containing $t$ and $t'$, respectively, we have $\mathcal{S}_{t'}$ $\subseteq$ $\mathcal{S}_t$. For each $q$-sequence $S\in \mathcal{S}_{t'}$, $t'$ = $t\diamondsuit_i w$, we will proof the theorem by proofing that $u(t',S) \le \textit{PEU}(t,S)$. Assuming $u(t',S)$ = $u(t'$, $\langle$ $j_1$, $j_2$, $\ldots$, $j_p\rangle, S)$, then there is always 
	\begin{equation}\label{eq1}
		u(t',S) = u(t,\ \langle j_1,j_2,\ldots,j_q\rangle,\ S) + u(w,\ \langle j_q,j_{q+1},\ldots,j_p\rangle,\ S).
	\end{equation}
	\begin{equation} \label{eq2}
		\textit{PEU}(t,\ j_q,\ S) = u(t,j_q,S) + u(S/_{I(t,j_q)})
	\end{equation}
	Observing the Eq. \ref{eq1} and Eq. \ref{eq2}, it's obviously that $u(t,j_q,S)$ $\ge$ $u(t$, $\langle$ $j_1$, $j_2$, $\ldots$, $j_q\rangle, S)$ and the $w \sqsubseteq $$S$/$_{I(t,j_q)}$, i.e., $ u(S/_{I(t,j_q)})$ $\ge$ $u(w$, $\langle$ $j_q$, $j_{q+1}$, $\ldots$, $j_p\rangle, S)$. Therefore, $u(t',S) \le \textit{PEU}(t,j_q,S)$. According to the Eq. \ref{eq3}, $\textit{PEU}(t,S)$ $\ge$ $\textit{PEU}(t,j_q,S)$. Thus, $u(t',S)$ $\le$ \textit{PEU}($t,S)$ and $\sum_{\forall S\in \mathcal{S}_{t'}} u(t',S)\le \sum_{\forall S\in \mathcal{S}_t'}\textit{PEU}(t,S)$ $\le$ \textit{PEU}($t, \mathcal{D}) $, i.e. $u(t', \mathcal{D})$ $\le$ \textit{PEU}($t, \mathcal{D})$. In the same way, $u(t',S)$ $\le$ \textit{PEU}($t,S$) when $t'$ = $t$ $\diamondsuit_s w$.
\end{proof}

Theorem \ref{PEU_theorem} indicates that, for any QSDB $\mathcal{D}$, the candidate sequences with low PEUs (less than $\xi \times u(\mathcal{D})$) can be pruned by HUSP-SP since the utilities of all these candidate sequences and their extension sequences will be no greater than the PEUs.

\begin{definition}(RSU \cite{wang2016efficiently}) \label{RSU}
	Let sequence $t'$ be generated by one item \textit{Extension} from sequence $t$. The \textit{RSU} of $t'$ in $q$-sequence $S$ is defined as
	\begin{equation}
	\textit{RSU}(t',S) = \begin{cases}
		\textit{PEU}(t,S), & t' \sqsubseteq S; \\
		0, & \text{otherwise}.
	\end{cases}
	\end{equation}
	Accordingly, the \textit{RSU} of the sequence $t'$ in the database $\mathcal{D}$ is defined as
	\begin{equation}
	\textit{RSU}(t', \mathcal{D}) = \sum_{\forall S \in \mathcal{D} } \textit{RSU}(t',S)
	\end{equation}
	Note, if $t'$ is a single item sequence, then $\textit{RSU}(t', \mathcal{D})$ is the same to the commonly known sequence-weighted utilization (SWU) \cite{yin2012uspan}, since $SWU(t')$ = $\sum_{S\in \mathcal{D}}\{u(S)|t' \subseteq S\}$.
	
\end{definition}

For example, in Table \ref{table1}, $\textit{RSU}(\langle \{b\}, \{e\} \rangle, \mathcal{D}  )$ = \textit{RSU}($\langle \{b\}, \{e\} \rangle,S_2$) + \textit{RSU}($\langle \{b\}, \{e\} \rangle,S_3$) = \textit{PEU}($\langle \{b\} \rangle,S_2$) + \textit{PEU}($\langle \{b\} \rangle,S_3)$ = 6 + 10 = 16.

\begin{theorem}\label{RSU_theo}
	For any \textit{Extension} of sequence $t$ in database $\mathcal{D}$ with sequence $w$, which can be empty, $t'$ = $t \diamondsuit w$, $u(t', \mathcal{D})$ $\le \textit{RSU}(t, \mathcal{D})$ \cite{yin2013efficiently}.
\end{theorem}

\begin{proof}
	Assuming $t$ is generated by one item \textit{Extension} from sequence $t_p$, $t$ = $t_p \diamondsuit i$, then $t'$ = $t_p \diamondsuit i \diamondsuit w$. Let $\mathcal{S}_{t_p}$ $\subseteq \mathcal{D}, \mathcal{S}_t$ $\subseteq \mathcal{D}$ and $ \mathcal{S}_{t'} \subseteq \mathcal{D}$ be the set of $q$-sequences containing $t_p$, $t$ and $t'$ respectively, we have $\mathcal{S}_{t'}$ $\subseteq \mathcal{S}_t \subseteq \mathcal{S}_{t_p}$. According to Definition \ref{RSU} and THEOREM \ref{PEU_theorem}, for each $q$-sequence $S \in \mathcal{S}_{t'}$, $\textit{RSU}(t,S)$ = $\textit{PEU}(t_p,S)$, $u(t',S) \le $ $\textit{PEU}(t_p,S)$. Hence, $u(t',S)$ $\le \textit{RSU}(t,S)$, $\sum_{\forall S\in \mathcal{S}_{t'}}u(t',S)\le \sum_{\forall S\in \mathcal{S}_t'}\textit{RSU}(t,S) \le \textit{RSU}(t, \mathcal{D})$, i.e., $u(t', \mathcal{D})\le \textit{RSU}(t, \mathcal{D})$.
\end{proof}

HUSP-SP can efficiently calculate the \textit{RSU} of the newly generated candidate pattern by adding the PEU values of $q$-sequences that still contain the new pattern. Then, according to Theorem \ref{RSU_theo}, HUSP-SP can prune the candidates with low \textit{RSU} values (less than $\xi \times u(\mathcal{D})$). However, \textit{RSU} only utilizes the \textit{PEU} of the prefix sequence for fast calculation, ignoring the help of the newly added item to reduce the utility upper bound. For instance, the sequence $\langle \{b\} \rangle$ generates $\langle \{b\}, \{e\} \rangle$ by one item extension, and $\textit{RSU}(\langle \{b\}, \{e\} \rangle, \mathcal{D})$ = \textit{PEU}($\langle \{b\} \rangle,S_2$) + \textit{PEU}($\langle \{b\} \rangle,S_3)$ = 6 + 10 = 16. \textit{PEU}($\langle \{b\} \rangle,S_3)$ = $u(\langle \{b\ d\},\{d\},\{a\ d\ e\}\rangle, S_3)$, where the subsequence between b and e, $\langle \{d\},\{d\},\{a\ d\}\rangle$, is irrelevant to the candidate pattern $\langle \{b\},\{e\}\rangle$ and $\langle \{b\},\{e\}\rangle$'s extended sequences but contributes to the utility upper bound of $\langle \{b\},\{e\}\rangle$. The same problem can be found in \textit{PEU}($\langle \{b\} \rangle,S_2$). Therefore, we designed a new tighter utility upper bound \textit{TRSU} to overcome the disadvantages of \textit{RSU} by considering the newly added item and subtracting the utility value of the irrelevant subsequence from the \textit{PEU} under certain conditions. We shall learn later that the value of $\textit{TRSU}(\langle \{b\}, \{e\} \rangle, \mathcal{D})$ is $7$, and it is much less than the value of $\textit{RSU}(\langle \{b\}, \{e\} \rangle, \mathcal{D})$, 16. 

\begin{definition}(TRSU)\label{TRSU}
	Assuming $q$-sequence $S$ of QSDB $\mathcal{D}$ has $k$ instances of sequence $t$ with extension positions $p_1$, $p_2$, $\ldots$, $p_k$, where $p_1$ $<$ $p_2 $ $<$ $\ldots$ $<$ $p_k$. Let sequence $t'$ be generated by one item \textit{Extension} from sequence $t$. The \textit{TRSU} of $t'$ in $q$-sequence $S$ is defined as
	\begin{equation}
	\textit{TRSU}(t',S) = \begin{cases}
		\textit{PEU}(t,S) - [u(S/_{I(t,p_i)}) - u(S/_{I(t',p'_1) - 1})], & t' \sqsubseteq S\text{, }\textit{PEU}(t,S) = u(t,\ p_1,\ S) + u(S/_{I(t,p_1)}); \\
		\textit{RSU}(t',S), & \text{otherwise}.
	\end{cases}
	\end{equation}
	where $p_1$ is the first extension position of $t$, $p'_1$ represents the first extension position of $t'$, and $p_i$ is the first extension position of $t$ before $p'_1$ ($p_i$ may equal to $p'_1$) in $S$. Also, the \textit{TRSU} of the sequence $t'$ in database $\mathcal{D}$ is defined as
	\begin{equation}
	\textit{TRSU}(t', \mathcal{D}) = \sum_{\forall S \in \mathcal{D} } \textit{TRSU}(t',S).
	\end{equation}		
\end{definition}

For example, in Table \ref{table1}, $\textit{TRSU}(\langle \{b\}, \{e\} \rangle, \mathcal{D})$ = $\textit{TRSU}(\langle \{b\}, \{e\} \rangle,S_2)$ + $\textit{TRSU}(\langle \{b\}, \{e\} \rangle,S_3 )$ = $ \textit{PEU}(\langle \{b\}\rangle,S_2)$ - $[u(S_2/_1)$ - $u(S_2/_3)]$ + $\textit{PEU}(\langle \{b\} \rangle,S_3)$ - $[u(S_3/2)$ - $u(S_3/6)]$ = 6 - (5 - 3) + 10 - (8 - 1) = 7. Besides, as the example in Definition \ref{RSU} states, the $\textit{RSU}(\langle \{b\}, \{e\} \rangle, \mathcal{D})$ = 16, which is even greater than twice the value of $\textit{TRSU}(\langle \{b\}, \{e\} \rangle, \mathcal{D})$. It shows the TRSU is much tighter than the RSU. Note that the TRSU and RSU are the same kinds of upper bounds that are designed based on the look-ahead strategy \cite{lin2017high}. Besides, TRSU provides an idea for further reducing the upper bounds designed based on the look-ahead strategy.

\begin{theorem}\label{TRSU_theo}
	For any \textit{Extension} of sequence $t$ in the database $\mathcal{D}$ with sequence $w$, $w$ can be empty, $t'$ = $t \diamondsuit w$, $u(t', \mathcal{D}) \le \textit{TRSU}(t, \mathcal{D})$.
\end{theorem}

\begin{proof}
	Assuming $t$ is generated by one item \textit{Extension} from sequence $t_p$, $t$ = $t_p \diamondsuit i$, then $t'$ = $t_p \diamondsuit i \diamondsuit w$. Let $\mathcal{S}_{t_p}$ $\subseteq \mathcal{D}$, $\mathcal{S}_t \subseteq \mathcal{D}$ and $\mathcal{S}_{t'}$ $\subseteq \mathcal{D}$ be the set of $q$-sequences containing $t_p$, $t$ and $t'$ respectively, we have $\mathcal{S}_{t'}$ $\subseteq \mathcal{S}_t$ $\subseteq \mathcal{S}_{t_p}$. According to Definition \ref{TRSU}, there are two scenarios for the computation of $\textit{TRSU}(t,S)$ for each $q$-sequence $S \in \mathcal{S}_{t'}$. On the one hand, when $\textit{PEU}(t_p,S)$ = $u(t_p,\ p_1,\ S)$ + $u(S/_{I(t_p,p_1)})$, $\textit{TRSU}(t,S)$ = $\textit{PEU}(t_p,S)$ - $[u(S/_{I(t_p,p_i)})$ - $u(S/_{I(t,p'_1) - 1})]$, where $p_1$ is the first extension position of $t_p$, $p'_1$ is the first extension position of $t$, and $p_i$ is the first extension position of $t_p$ before $p'_1$. According to THEOREM \ref{RSU_theo}, $u(t',S) \le \textit{PEU}(t_p,S)$, and it can be observed that the sub $q$-sequence of $S$ from index $I(t_p,p_i) + 1$ to $I(t,p'_1) - 1$ is irrelevant to any $t'$ (there will be no subsequence of $t'$ contained in this sub $q$-sequence). Therefore, $u(t',S)$ $\le$ $\textit{PEU}(t_p,S)-[u(S/_{I(t_p,p_i)})$ - $u(S/_{I(t,p'_1) - 1})]$, since the $u(S/_{I(t_p,p_1)})$ contains the utility of the irrelevant sub $q$-sequence, that's to say, $u(t',S) \le $ \textit{TRSU}($t,S$). On the other hand, \textit{TRSU}($t,S$) = \textit{RSU}($t,S$) = \textit{PEU}($t_p,S$), and we have $u(t',S)$ $\le$ $\textit{PEU}(t_p,S)$. Hence, $u(t',S)$ $\le $ \textit{TRSU}($t,S$) also holds. In conclusion, $u(t',S)$ $\le $ \textit{TRSU}($t,S$), $\sum_{\forall S\in \mathcal{S}_{t'}}u(t',S)\le \sum_{\forall S\in \mathcal{S}_{t'}}\textit{TRSU}(t,S)$ $\le \textit{TRSU}(t, \mathcal{D})$, i.e., $u(t', \mathcal{D})\le \textit{TRSU}(t, \mathcal{D})$.
\end{proof}

HUSP-SP can also efficiently calculate the \textit{TRSU} of newly generated candidate patterns with the help of the \textit{remaining-utility array}. According to THEOREM \ref{TRSU_theo}, HUSP-SP can prune the candidates with low TRSU values (less than $\xi$ $\times$ $u(\mathcal{D})$). Based on the utility upper bounds, the pruning strategies are described as follows. First, the IIP (irrelevant items pruning) strategy \cite{gan2020fast} is adopted by HUSP-SP, and it is defined based on RSU in this paper as follows. \textit{IIP Strategy:} Given a sequence $t$, and any item $i$ available for extension, if $\sum_{S\in \mathcal{D}}$ \textit{RSU}($t \diamondsuit i)$ is less than the minimum utility threshold ($\xi$ $\times$ $u(\mathcal{D})$), then $i$ can be removed from the seqPro structure of $t$ and $t$'s extension sequences. The correctness of the IIP strategy and proof process have been given in HUSP-ULL \cite{gan2020fast}. Then, according to THEOREM \ref{TRSU_theo}, an EP (early pruning) strategy is proposed to discard unpromising candidate items early. \textit{EP Strategy:} Given a sequence $t$ and any item $i$ available for extension, two situations are considered: 1). If $i$ is an I-Extension candidate item, and \textit{TRSU}($t \diamondsuit_i, \mathcal{D})$ $<$ $\xi$ $\times$ $u(\mathcal{D})$, then $i$ should be discarded. 2). If $i$ is an S-Extension candidate item, and \textit{TRSU}($t \diamondsuit_s i, \mathcal{D}) $ $<$ $\xi$ $\times$ $u(\mathcal{D})$, then $i$ should be discarded.

\subsection{HUSP-SP Algorithm}

Based on the seq-array, the designed TRSU, and the EP strategy, the proposed HUSP-SP algorithm can be stated as follows. Algorithm \ref{main} gives the main steps of HUSP-SP. The algorithm first scans the quantitative sequential database $\mathcal{D}$ to construct the storage structure, that is, the seq-array of each $q$-sequence $S \in \mathcal{D}$ (line 1). It also accumulates the utility of the $q$-sequences and gets the $u(\mathcal{D})$ after scanning the database (line 1). Then, the projected database \textit{seqPro}($\langle \rangle$) of the empty sequence is constructed (lines 3-4). Following that, the empty sequence is treated as the prefix, and HUSP-SP begins the depth-first search with the built projected database by invoking the \textbf{PatternGrowth} procedure (line 5).

\begin{algorithm}[h]
	\caption{HUSP-SP algorithm}
	\label{main}
	\KwIn{
		$\mathcal{D}$: a  quantitative sequential database;	\textit{UT}: a utility table with external utility values for distinct items in $\mathcal{D}$; $\xi$: a minimum utility threshold.}
	\KwOut{
		the set of \textit{HUSPs}.}
	\BlankLine
	scan $\mathcal{D}$ to build the seq-array for each $S \in \mathcal{D}$, and calculate $u(\mathcal{D})$\;
	\textit{seqPro}($\langle \rangle$).\textit{seq-array} $ \leftarrow \{\text{seq-array of } S\text{ | } S \in \mathcal{D}\}$\;
	\textit{seqPro}($\langle \rangle$).\textit{exList} $\leftarrow$ \textit{NULL}\;
	\textit{HUSPs} $\leftarrow \varnothing$\;
	call \textbf{PatternGrowth}($\langle \rangle$, \textit{seqPro}($\langle \rangle$), \textit{HUSPs})\;
	\Return{\textit{HUSPs}}
	
\end{algorithm}

\begin{algorithm}[h]
	\caption{PatternGrowth(\textit{prefix}, \textit{seqPro(prefix)}, \textit{HUSPs})}
	\label{pgrowth}
	\textit{seqPro(prefix)} $\leftarrow$ \textit{IIP}\textit{(seqPro}(\textit{prefix})) \tcp{measured by IIP}
	scan \textit{seqPro}(\textit{prefix}) to get I-Extension items (iList) and S-Extension items (\textit{sList}) of \textit{prefix} \tcp{measured by EP}
	\For{\rm each item $i$ $\in$ \textit{iList}}{
		call \textbf{UtilityCalculation}(\textit{prefix}$\diamondsuit_i i$, \textit{seqPro}(\textit{prefix}), \textit{HUSPs});
	}
	\For{\rm each item $i\in sList$ }{
		call \textbf{UtilityCalculation}(\textit{prefix}$\diamondsuit_s i$, \textit{seqPro}(\textit{prefix}), \textit{HUSPs});
	}
\end{algorithm}

\begin{algorithm}[h]
	\caption{UtilityCalculation(\textit{prefix}$'$, \textit{seqPro}(\textit{prefix}), \textit{HUSPs})}
	\label{utilComput}
	\textit{seqPro(prefix$'$).seq-array} $\leftarrow \{\text{seq-array of } S\text{ | } \textit{prefix}'\subseteq$ $S \land$ $S \in $ \textit{seqPro(prefix).seq-array}\}\;
	calculate $u$(\textit{prefix}$'$), \textit{PEU}(\textit{prefix}$'$) and construct \textit{seqPro}(\textit{prefix}$'$).\textit{exList}\;
	\If{$u$(\textit{prefix}$'$) $\ge$ $\xi$ $\times$ $u(\mathcal{D})$}{
		\textit{HUSPs} $\leftarrow $ \textit{HUSPs} $\cup $ \textit{prefix}$'$\;
	}
	\If{\textit{PEU}(\textit{prefix}$'$) $\ge$ $\xi$ $\times$ $u(\mathcal{D})$}{
		call \textbf{PatternGrowth}(\textit{prefix}$'$, \textit{seqPro(prefix$'$)}, \textit{HUSPs})\;
	}
\end{algorithm}

The \textbf{PatternGrowth} procedure (cf. Algorithm \ref{pgrowth}) shows the depth-first search process, i.e., the pattern growth process through the I-Extension and S-Extension operations. The algorithm first removes the irrelevant items from the projected database by applying the IIP strategy, and then updates the projected database, \textit{seqPro}(\textit{prefix}) (line 1). Then, the reduced projected database \textit{seqPro}(\textit{prefix}) is scanned to get the promising I-Extension items (\textit{iList}) and S-Extension items (\textit{sList}) of \textit{prefix} by EP strategy (line 2). For each item $i$ of iList and sList, the algorithm generates a new one length longer pattern by applying I-Extension and S-Extension with $i$, respectively. Furthermore, the newly generated pattern, such as \textit{prefix} $\diamondsuit_s i$, is evaluated by calling the \textbf{UtilityCalculation} procedure (lines 3-8).

The \textbf{UtilityCalculation} procedure (cf. Algorithm \ref{utilComput}) first establishes the seq-array field of the projected database \textit{seqPro}(\textit{prefix}$'$) based on the \textit{seqPro}(\textit{prefix}) (line 1). Then, the utility and PEU value of \textit{prefix}$'$ are calculated based on the constructed \textit{seqPro}(\textit{prefix}$'$); meanwhile, the Extension-List field of the \textit{seqPro}(\textit{prefix}$'$) is built (line 2). If the $u$(\textit{prefix}$'$) is not less than the minimum utility $\xi$ $\times$ $u(\mathcal{D})$, the \textit{prefix}$'$ is added to the \textit{HUSPs} (lines 3-5). Furthermore, if the \textit{PEU(prefix$'$)} is not less than $\xi$ $\times$ $u(\mathcal{D})$, the \textbf{PatternGrowth} procedure is called to mine the HUSPs prefixed with \textit{prefix}$'$ (lines 6-8). Note, the utilization of PEU can result in missing patterns when the pruning condition ``if \textit{PEU(prefix$'$)} $<$ $\xi$ $\times$ $u(\mathcal{D})$ then stop the mining branch with \textit{prefix}$'$ as parent node" is executed before displaying the result sequence ``if ($u(\textit{prefix}')$ $\ge$ $\xi$ $\times$ $u(\mathcal{D})$) then output \textit{prefix}$'$" \cite{truong2019survey}. Finally, HUSP-SP terminates when there are no newly generated candidate patterns, and returns the set of HUSP, \textit{HUSPs}.

To better illustrate the proposed algorithm, we give a running example below. Considering the running example (Table \ref{table1} and Table \ref{table2}), and $\xi$ = 0.5. After the first scanning of the database, we build the seq-arrays for $q$-sequences of the database and obtain that the threshold value is $\xi$ $\times$ $u(\mathcal{D})$ = 23.5, the \textit{SWU} \cite{yin2012uspan} value of $\langle \{a\} \rangle$, $\langle \{b\} \rangle$, $\langle \{c\} \rangle$, $\langle \{d\} \rangle$, $\langle \{e\} \rangle$, $\langle \{f\} \rangle$ are 29, 35, 12, 47, 34, 31. Thus, the item $c$ is permanently deleted from the database as its \textit{SWU} value is less than the threshold value. Then we get the \textit{TRSU} value of $\langle \{a\} \rangle$, $\langle \{b\} \rangle$, $\langle \{d\} \rangle$, $\langle \{e\} \rangle$, $\langle \{f\} \rangle$ are 29, 23, 22, 10, and 10 by scanning the original seq-arrays. So, $\langle \{a\} \rangle$ is the only promising candidate. Next, we scan the original seq-arrays again to construct the projected database (\textit{seqPros}) and calculate the utility of $\langle \{a\} \rangle$. The utility and \textit{PEU} values of $\langle \{a\} \rangle$ are $12$ and $29$, separately. Therefore, $\langle \{a\} \rangle$ is not a HUSP, and we should call the \textbf{PatternGrowth} on it for mining its extension sequences. Firstly, we get the \textit{RSU} value of $\langle \{a \diamondsuit e\} \rangle$ and $\langle \{a \diamondsuit f\} \rangle$ are 16 and 13 by searching the \textit{seqPro} of $\langle \{a\} \rangle$. Thus, we delete the items $e$ and $f$ from the \textit{seqPro} of $\langle \{a\} \rangle$ based on the \textit{IIP strategy}. Then, we scan the reduced \textit{seqPro}($\langle \{a\} \rangle$), and we find two \textit{I-Extension} sequences $\langle \{a, b\} \rangle$ and $\langle \{a, d\} \rangle$, whose \textit{TRSU} value are $27$ and $25$. The \textit{S-Extension} sequences of $\langle \{a\} \rangle$ include $\langle \{a\}, \{a\} \rangle$ and $\langle \{a\}, \{d\} \rangle$, whose \textit{TRSU} value are $21$ and $24$. According to the \textit{EP strategy}, we can discard the $\langle \{a\}, \{d\} \rangle$ as its \textit{TRSU} value is less than 23.5. $\langle \{a, b\} \rangle$, $\langle \{a, d\} \rangle$ and $\langle \{a\}, \{d\} \rangle$ are the promising candidates, and we call the \textbf{UtilityCalculation} for each of them. The mining process for $\langle \{a, d\} \rangle$ and $\langle \{a\}, \{d\} \rangle$ is terminated since their \textit{PEU} values are 17 and 19. The utility and \textit{PEU} values of $\langle \{a, b\} \rangle$ are 16 and 27, so $\langle \{a, b\} \rangle$ is not a HUSP. Finally, we call the \textbf{PatternGrowth} on $\langle \{a, b\} \rangle$ and get one HUSP $\langle \{a, b\}, \{a, d\}\rangle$ with a utility of 25.

\subsection{Complexity Analysis} \label{compx}

Let $|\mathcal{D}|$ denotes the number of $q$-sequences in database $\mathcal{D}$, $L$ denotes the length of the longest $q$-sequence in $\mathcal{D}$, and $|I|$ denotes the number of distinct items in $\mathcal{D}$. HUSP-SP first scans the database $\mathcal{D}$ to build the seq-arrays, whose time complexity is $O(|\mathcal{D}|L)$. The memory complexity of the seq-arrays is also $O(|\mathcal{D}|L)$ since there are $|\mathcal{D}|$ seq-arrays, and the most extended length is $L$. Then, HUSP-SP calls the recursive function \textbf{PatternGrowth} after some initial operations and finally returns the \textit{HUSPs}. Thus, the time complexity of the HUSP-SP algorithm is $O(|\mathcal{D}|L)$ + $N$ $\times$ $\textit{timeComp}\_PG$, where $N$ denotes the number of times \textbf{PatternGrowth} is called, and $\textit{timeComp}\_PG$ denotes the time complexity of \textbf{PatternGrowth}. The memory complexity of HUSP-SP is $O(|\mathcal{D}|L)$ + $H$ $\times$ $memoComp\_PG$, where H denotes the maximum depth of recursively calling \textbf{PatternGrowth}, $\textit{memoComp}\_PG$ denotes the memory complexity of \textbf{PatternGrowth}.

As for \textbf{PatternGrowth}, the procedure first needs to scan the projected database three times (lines 1-2), whose worst time complexity is $O(|\mathcal{D}|L)$. Specifically, the IIP operation first takes one scan to mark the extension items with low \textit{RSU} as irrelevant items, and then scans the projected database again to update the projected database by deleting the utility of the irrelevant items in the \textit{Remaining-utility array} (Definition \ref{seqA}). After the IIP operation, \textbf{PatternGrowth} scans the reduced projected database once to get the promising extension items, \textit{iList} and \textit{sList}. The remaining operations of \textbf{PatternGrowth} are appending each item of \textit{iList} or \textit{sList} to the prefix and calculating the utility of the generated candidate pattern (lines 3-8), whose worst time complexity is $|I|$ $\times$ $\textit{timeComp}\_UC$. Note that the $\textit{timeComp}\_UC$ represents the time complexity of \textbf{UtilityCalculation}. Thus, the worst time complexity of \textbf{PatternGrowth} is $O(|\mathcal{D}|L)$ + $|I|$ $\times$ $\textit{timeComp}\_UC$. In addition, the worst memory complexity of \textbf{PatternGrowth} is $O(|\mathcal{D}|L$ + $|I|)$. This is because the largest size of the projected database can be $|\mathcal{D}|$ when each $q$-sequence in $\mathcal{D}$ contains the candidate pattern. Moreover, the largest size of \textit{Extension-list} can be $L$ when the candidate pattern is equal to any subsequence of the same length in $\mathcal{D}$. Therefore, the worst memory complexity of the projected database can be $O(|\mathcal{D}|L)$. In addition, for marking deleted items, the IIP operation requires a global array of length $|I|$.

Additionally, the worst time complexity of the \textbf{UtilityCalculation} is $O(|\mathcal{D}|L)$. This is because the projected database may contain at most $|\mathcal{D}|$ $q$-sequences. Besides, for each $q$-sequence, calculating the utility of the new candidate pattern needs to traverse the extension-List and Item-indices table, which may contain at most $L$ elements. According to the above, the value of \textit{timeComp\_UC} is $O(|\mathcal{D}|L)$, and the value of \textit{timeComp\_PG} is $O(|\mathcal{D}|L)$ + $|I|$ $\times$ $O(|\mathcal{D}|L)$, which equals $O(|I||\mathcal{D}|L)$. Thus, the time complexity of the HUSP-SP algorithm is $O(|\mathcal{D}|L)$ + $N $ $\times$ $O(|I||\mathcal{D}|L)$, which equals $O(N|I||\mathcal{D}|L)$. Besides, the value of \textit{memoComp\_PG} is $O(|\mathcal{D}|L + |I|)$, and the memory complexity of HUSP-SP is $O(|\mathcal{D}|L)$ + $H$ $\times$ $O(|\mathcal{D}|L)$ + $ O(|I|)$, which equals $O(H|\mathcal{D}|L + |I|)$. Note that there is only one global array of length $|I|$ for marking deleted items; thus, the factor of $|I|$ is one instead of $H$. To sum up, the time complexity of the HUSP-SP algorithm is $O(N|I||\mathcal{D}|L)$, and the memory complexity is $O(H|\mathcal{D}|L + |I|)$.

It is worth noting that the number of candidate patterns (LQS-tree nodes) is $O(|I|^L)$, which is equal to the value of $N$. Furthermore, the length of the longest pattern is $L$, which equals the value of $H$. Therefore, the worst time complexity and memory complexity of the HUSP-SP algorithm are $O(|I|^2|\mathcal{D}|L^2)$ and $O(|\mathcal{D}|L^2 + |I|)$, separately. Nevertheless, the exact time and memory complexity of HUSP-SP can be much smaller than the above theoretical values, as the proposed pruning strategies can significantly reduce the search space (the number of candidate patterns).

\section{Experiments}   \label{sec:experiments}

In this section, sufficient experimental results were presented and analyzed to demonstrate the performance of the proposed algorithm. The state-of-the-art HUSPM algorithms, including USpan \cite{yin2012uspan} (replaced the SPU by PEU), ProUM \cite{gan2020proum} and HUSP-ULL \cite{gan2020fast} were selected as the baselines. All the compared algorithms were implemented in Java. All the course code and datasets are available at GitHub\footnote{https://github.com/DSI-Lab1/HUSPM}. The following experiments were conducted on a personal computer with an Intel Core i7-8700K CPU @ 3.20 GHz, a 3.19 GHz processor, 8 GB of RAM, and a 64-bit Windows 10 operating system.

\subsection{Data Description}

Five real-world datasets and one synthetic dataset were utilized in our experiment to evaluate the performance of the compared algorithms. Table \ref{features} lists the statistical characteristics of these datasets. Note that the number of $q$-sequences is denoted as $|D|$, the number of different $q$-items is denoted as $|I|$, the average/maximum length of $q$-sequences is denoted as $\textit{avg}(S)$/$\textit{max}(S)$, the average number of $q$-itemsets per $q$-sequence is denoted as \textit{\#avg(IS)}, the average number of $q$-items per $q$-itemset is denoted as \textit{\#Ele}, and the average/maximum utility of $q$-items is denoted as \textit{avg}(\textit{UI})/\textit{max}(\textit{UI}). The values of \textit{\#Ele} parameter of these datasets indicate that the Sign, Bible, Kosarak10k and Leviathan are composed of single item element based sequences, while the Yoochoose and SynDataset-160K are composed of multi-item element based sequences. The utilization of both single-item and multi-item element based sequence datasets makes the experimental results more convincing. Excluding \textit{Yoochoose}\footnote{https://recsys.acm.org/recsys15/challenge}, the other datasets can be obtained from an open-source data mining website\footnote{http://fimi.ua.ac.be/data}.

\begin{table}[h]
	\centering  
	\caption{Features of the datasets}
	\label{features}
    \begin{tabular}{lllllll}
    	\hline
    	\textbf{Dataset} & \textbf{$|\textit{D}|$} & \textbf{$|\textit{I}|$} & \textbf{$\textit{avg}(\textit{S})$} & \textbf{$\textit{max}(\textit{S})$} & \textbf{\textit{avg}(\textit{IS})} & \textbf{\textit{\#Ele}} \\ 
    	\hline 
    	Sign & 730 & 267 & 52.00 & 94 & 52.00 & 1.00  \\ 
    	Bible & 36,369 & 13,905 & 21.64 & 100 & 21.64 & 1.00  \\ 
    	SynDataset-160k & 159,501 & 7,609 & 6.19 & 20 & 26.64 & 4.32  \\ 
    	Kosarak10k & 10,000 & 10,094 & 8.14 & 608 & 8.14 & 1.00  \\ 
    	Leviathan & 5,834 & 9,025 & 33.81 & 100 & 33.81 & 1.00  \\ 
    	Yoochoose & 234,300 & 16,004 & 2.25 & 112 & 1.14 & 1.98  \\ 
        \hline
    	SynDataset-10k  & 10,000  & 7,312 & 27.11 & 213 & 6.23 & 4.35  \\ 
        SynDataset-80k  & 79,718  & 7,584 & 26.80 & 213 & 6.20 & 4.32  \\ 
        SynDataset-160k & 159,501 & 7,609 & 26.75 & 213 & 6.19 & 4.32  \\ 
        SynDataset-240k & 239,211 & 7,617 & 26.77 & 213 & 6.19 & 4.32  \\ 
        SynDataset-320k & 318,889 & 7,620 & 26.76 & 213 & 6.19 & 4.32  \\ 
        SynDataset-400k & 398,716 & 7,621 & 26.75 & 213 & 6.19 & 4.32 \\ 
        \hline 
	\end{tabular}
\end{table}

\subsection{Efficiency Analysis}

In the first experiment, the runtime performance of the proposed algorithm was compared with the state-of-the-art algorithms. Then, on six datasets, a series of experiments with various minimum utility threshold (denoted as $\xi$) settings were run, with the detailed results shown in Fig. \ref{runtime}.

\begin{figure*}[h]
	\centering
	\includegraphics[width=1\linewidth]{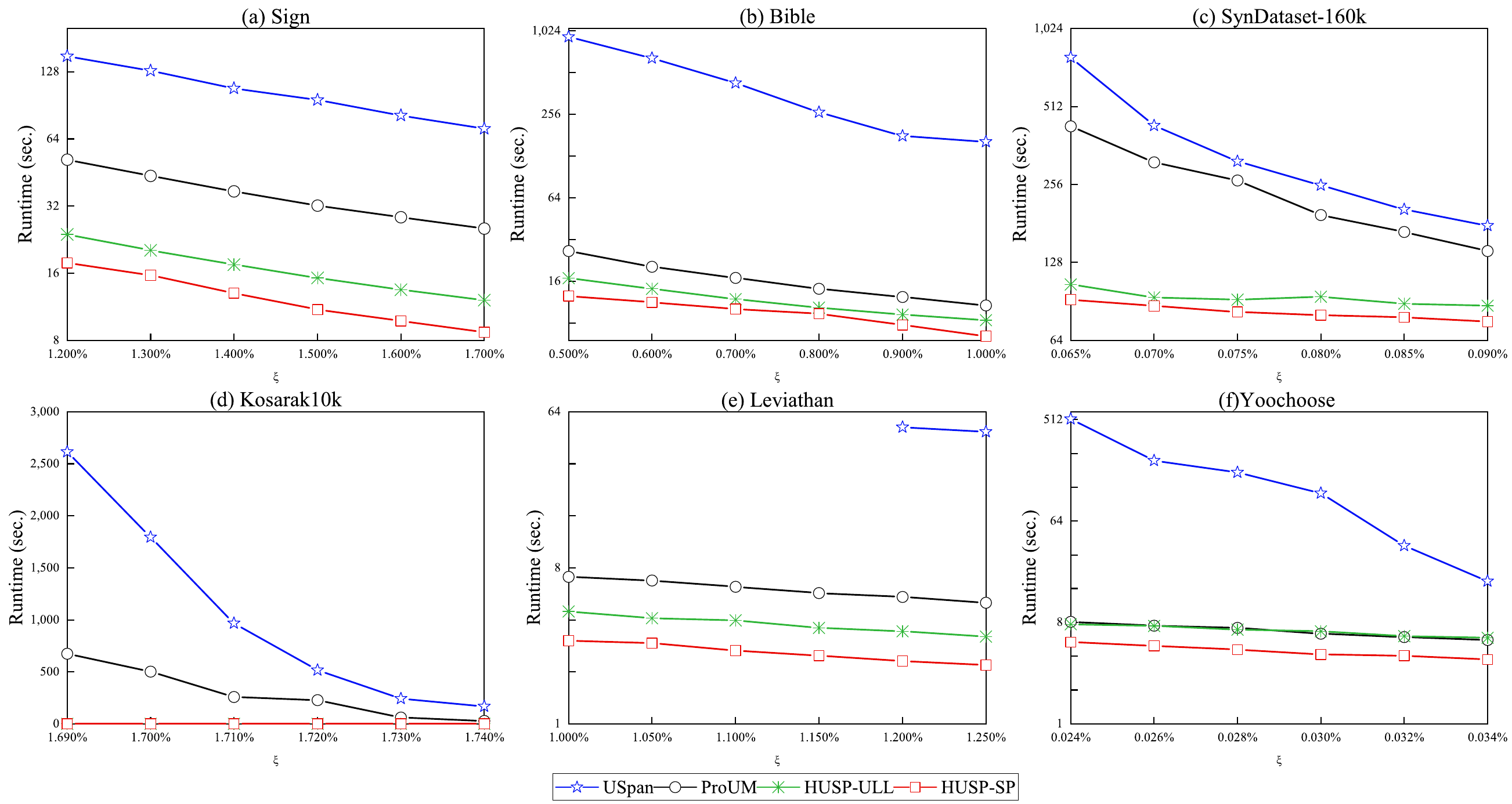}
	\caption{Execution times of the compared methods under various minimum utility thresholds}
	\label{runtime}
\end{figure*}

It is shown that the proposed HUSP-SP is faster than the other existing algorithms in all cases. Generally, HUSP-SP was a quarter faster than the state-of-the-art algorithm HUSP-ULL, one order of magnitude faster than ProUM, and two or three orders of magnitude faster than USpan. For example, except for the SynDataset-160k in Fig. \ref{runtime}(c), the runtime of HUSP-SP consumes around ten seconds, while the other algorithms may consume hundreds to thousands of seconds. Furthermore, HUSP-SP and HUSP-ULL are much more stable than ProUM and USpan in runtime efficiency when $\xi$ decreases. For instance, in Fig. \ref{runtime}(c), the runtime of USpan and ProUM increased by a hundred seconds when $\xi$ decreased by 0.00005, while HUSP-SP ran only a few seconds longer. Note that when the minimum threshold parameter is less than $1.2\%$ on the Leviathan dataset, the USpan algorithm cannot finish the experiment because it has run out of memory. In all parameter settings, HUSP-SP is superior to all existing algorithms.

\subsection{Effectiveness of Pruning Strategies}

In order to evaluate the effect of pruning strategies, this subsection investigated the generated candidate patterns and HUSPs of the compared algorithms under varying minimum utility thresholds in the six datasets. The details are shown in Fig. \ref{candidate}. Note that \textit{\#candidate} is the number of the generated candidate patterns, which must be checked, and \textit{\#HUSPs} is the number of the final HUSPs discovered by the method.

\begin{figure*}[h]
	\centering
	\includegraphics[width=1\linewidth]{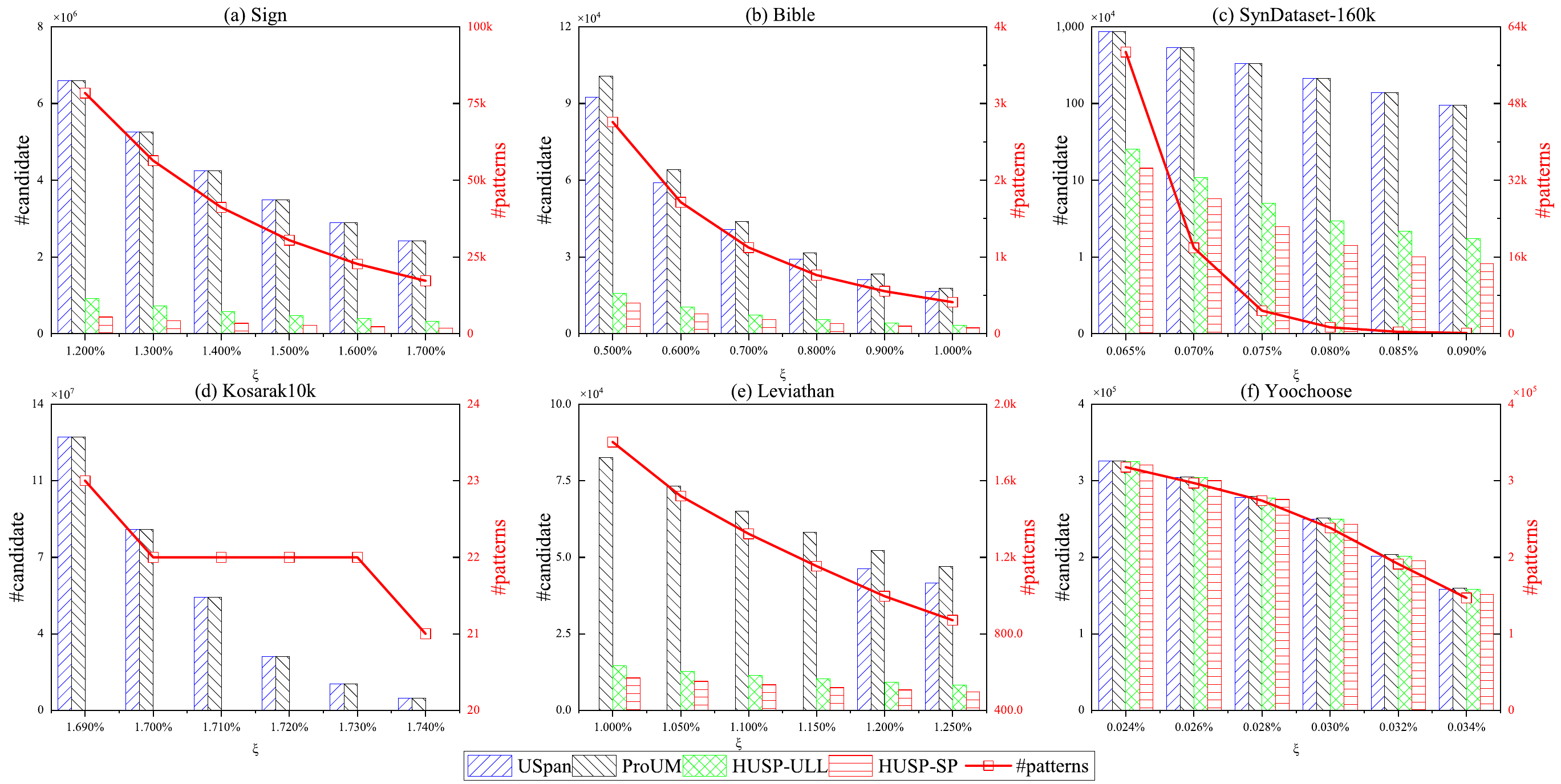}
	\caption{Experimental results of generated candidates and patterns}
	\label{candidate}
\end{figure*}

As shown in Fig. \ref{candidate}, the number of candidate patterns generated by HUSP-ULL and HUSP-SP is much smaller than the other two methods. The main reason for the huge difference in the number of candidates is that HUSP-ULL and HUSP-SP adopt the IIP strategy. With the introduction of the IIP strategy, the utility upper bounds decreased faster due to the utility deletion of irrelevant items. Therefore, the search space was much smaller when the IIP strategy worked. Besides, HUSP-SP generally generates half as many candidate patterns as HUSP-ULL. Therefore, it proves that the proposed new upper bound TRSU and the EP pruning strategy are effective.

It can also be observed that the number of candidate patterns of HUSP-SP and HUSP-ULL increased much slower than USpan and ProUM. However, Fig. \ref{candidate}(f) shows that, in the dataset \textit{Yoochoose}, the compared algorithms generated a similar number of candidate patterns. Still, from Fig. \ref{runtime} and Fig. \ref{memory}, we can find that the proposed HUSP-SP performed better in terms of runtime and memory. Therefore, the proposed projected structure \textit{seqPro} is more compact and more effective in the HUSP mining process.

\subsection{Memory Evaluation}

In this subsection, the important algorithmic measure criteria for memory are evaluated. The experiment results are shown in Fig. \ref{memory}.

\begin{figure*}[h]
	\centering
	\includegraphics[width=1\linewidth]{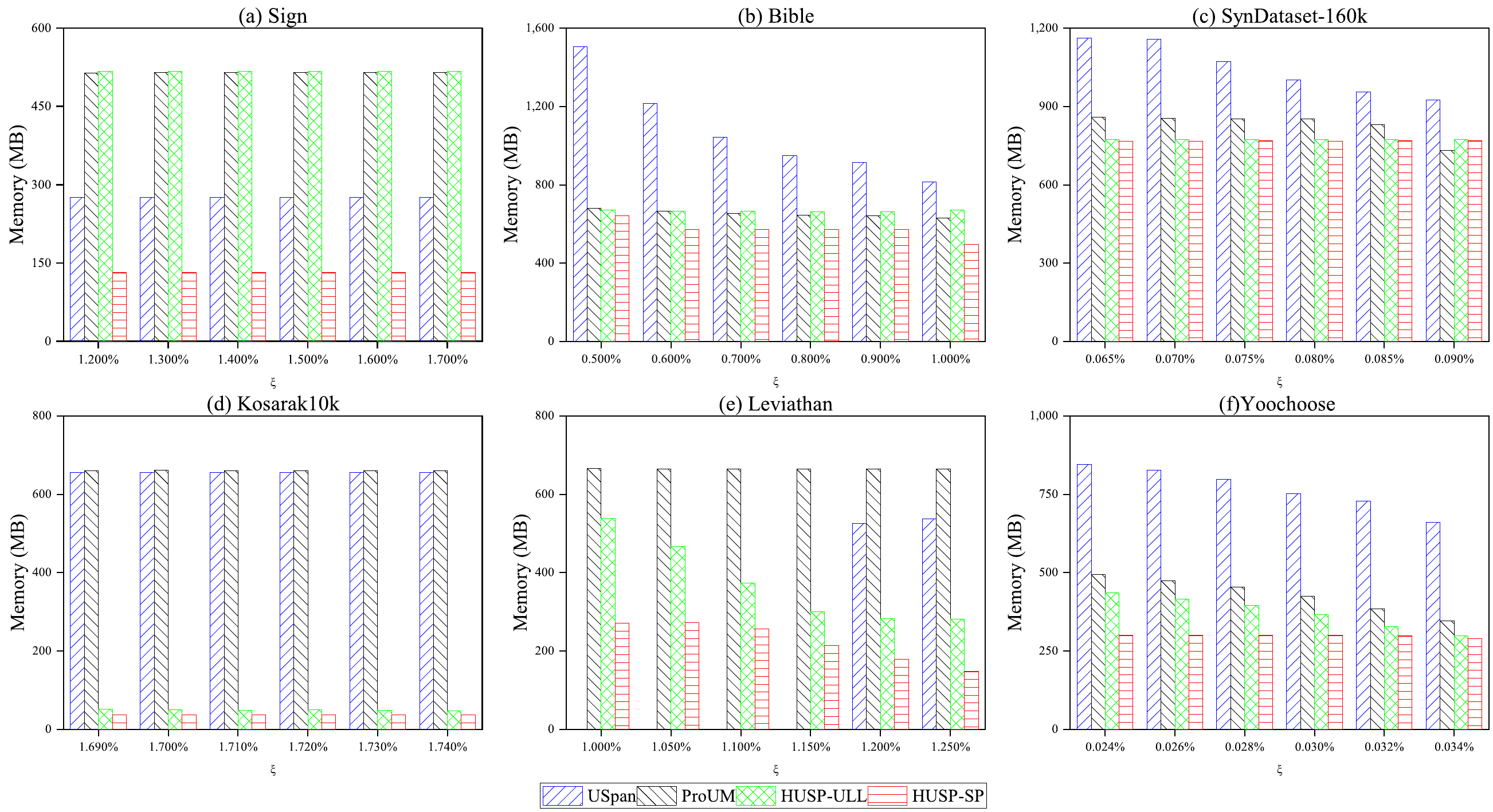}
	\caption{Memory usage results of the compared methods under various minimum utility thresholds}
	\label{memory}
\end{figure*}

HUSP-SP outperformed all of the compared algorithms in terms of memory consumption across all experiment parameter settings, as shown in Fig. \ref{memory}. It can be observed that the memory usage of the HUSPM method increased when the number of candidate patterns increased. For example, Fig. \ref{memory}(d) shows that USpan (ProUM) consumed about 600 megabytes more memory than HUSP-SP (HUSP-ULL), while USpan (ProUM) generated over 100,000,000 more candidate patterns than HUSP-SP (HUSP-ULL). The memory performance of USpan is generally poor. For instance, USpan ran out of memory when the minimum utility threshold $\xi$ was less than 1.20\% in Leviathan. Also, it can be found that the data structures utilized by ProUM and HUSP-ULL are not compact enough in some conditions. For example, in Fig. \ref{memory}(a), the USpan consumed about 200 megabytes less memory than ProUM and HUSP-ULL. However, the performance of the utility-matrix \cite{yin2012uspan} utilized by USpan was poor. In conclusion, the good memory usage performance of HUSP-SP proves that the newly proposed seq-array structure is compact, and the search space of HUSP-SP is much smaller.

\subsection{Scalability Test}

The robustness of the compared algorithms is analyzed in this subsection through the scalability test. The experiment was based on a synthetic multi-item element-based sequence dataset, namely C8S6T4I3D|X|K \cite{r1994quest}. The detailed results are shown in Fig. \ref{scalability}, including runtime, candidate, and memory efficiency. Note that the size of the \textit{SynDataset} varied from 10K to 400K sequences, and the minimum utility threshold $\xi$ was set to 0.001 throughout the experiment.

\begin{figure*}[h]
	\centering
	\includegraphics[width=1\linewidth]{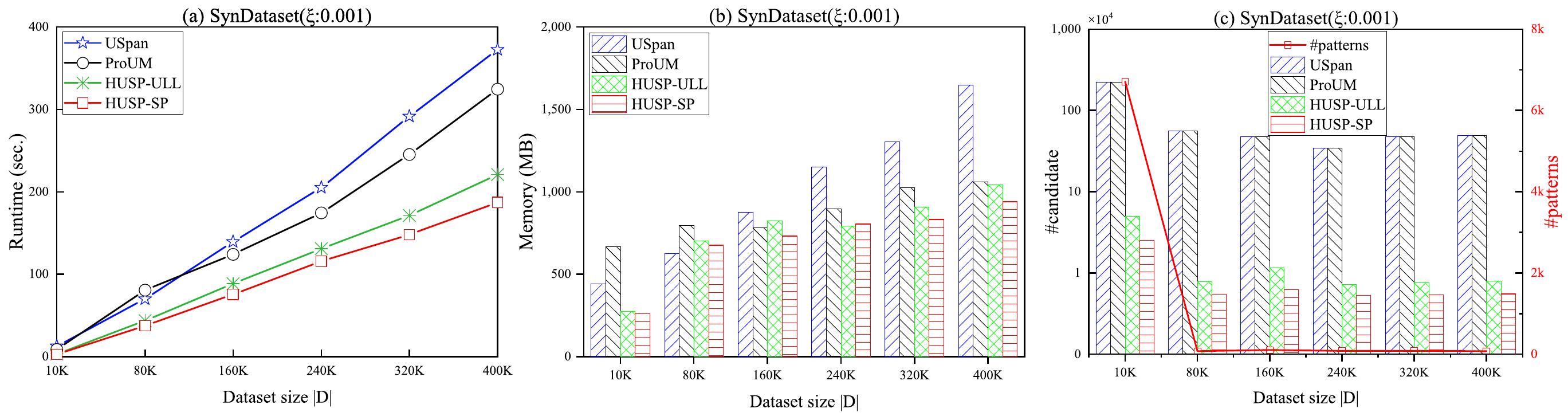}
	\caption{Scalability of the compared methods}
	\label{scalability}
\end{figure*}

As shown in Fig. \ref{scalability}, HUSP-SP had the best scalability among the compared algorithms for its minimum runtime, memory consumption, and candidate pattern number in all the test results. Furthermore, the runtime of HUSP-SP increased linearly as the number of dataset sequences grew. From Fig. \ref{scalability}(b), it can be observed that the memory usage of HUSP-SP is relatively stable, excluding the experiment with 10K sequences. Similar results can be found in Fig. \ref{scalability}(c) where the number of candidate patterns was kept stable while the size of \textit{SynDataset} varied from 80K to 400K. Therefore, the growth in memory usage of HUSP-SP comes mainly from the expansion of the processed dataset. The proposed HUSP-SP algorithm has good extensibility for dealing with large-scale datasets.

\subsection{Ablation Study}

We further conducted an ablation study on the proposed upper bound TRSU to evaluate its effects on execution performance in terms of runtime, memory consumption, and the number of generated candidates. Theoretically, introducing TRSU can reduce the search space and speed up the mining process. To evaluate the effectiveness of TRSU, we developed another algorithm, HUSP-SP*, by replacing the upper bound of HUSP-SP from TRSU to RSU. We examined two datasets, SynDataset80k and SynDataset160k, with different minimum utility threshold $\xi$ settings. The detailed study results are shown in Fig. \ref{ablation}.

\begin{figure*}[h]
	\centering
	\includegraphics[width=1\linewidth]{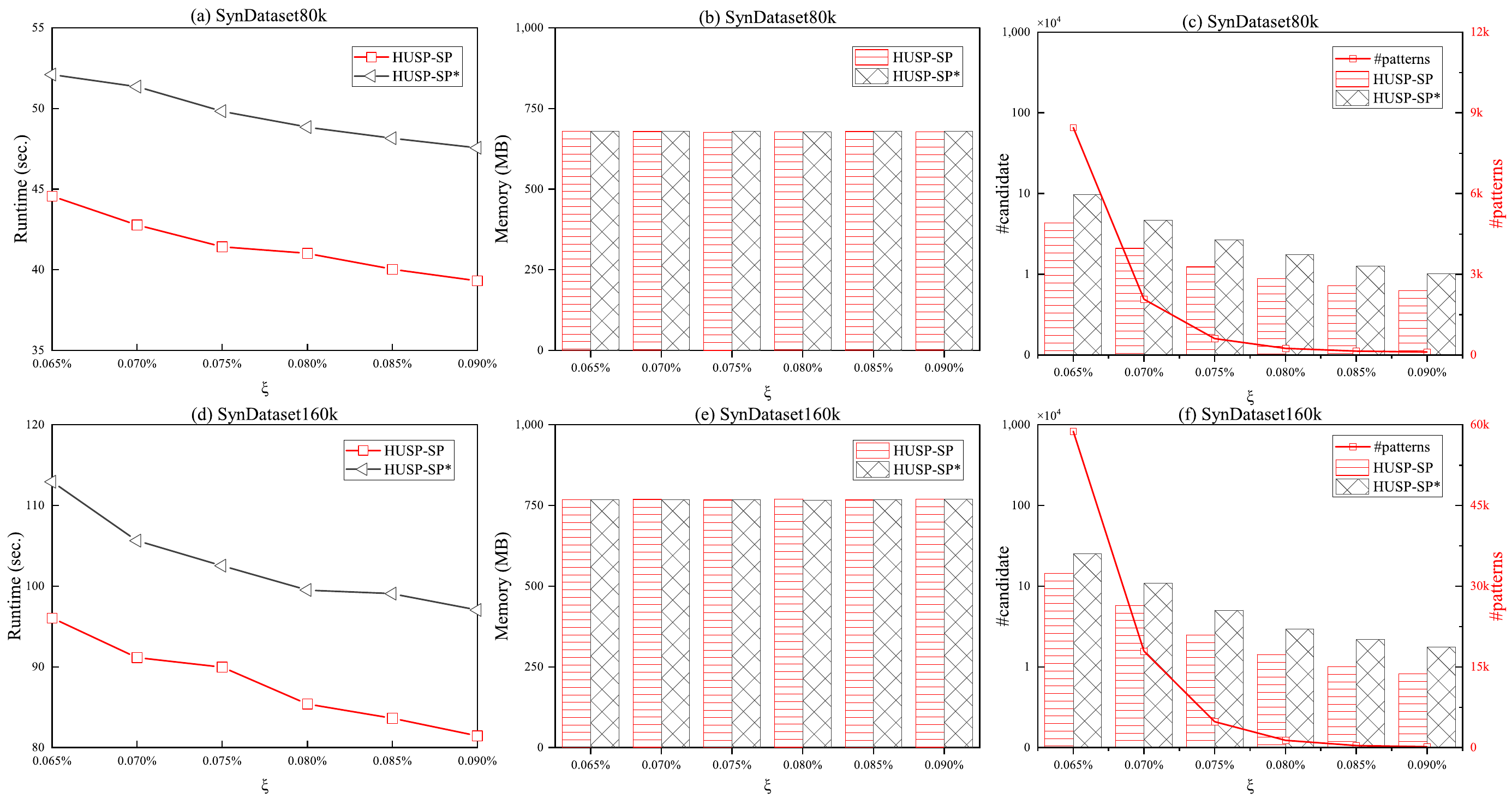}
	\caption{Ablation results of the TRSU}
	\label{ablation}
\end{figure*}

As we can see, HUSP-SP has a shorter runtime and fewer candidate patterns in all cases of the study results. It proves that TRSU does contribute to the high performance of HUSP-SP. It is interesting that HUSP-SP and HUSP-SP* have the same memory consumption result in all cases. To explain this phenomenon, we perform a further experiment and find that when we continually change the minimum utility threshold, e.g., set to 0.0004, the memory consumption of HUSP-SP on SynDataset80K increases to 750 MB, and the candidate number increases to 7,888,232. That is to say, the number of candidates increased by nearly 8,000,000 while memory consumption increased by less than 100 MB. Besides, in all the study cases in Fig. \ref{ablation}, we can find that the number of candidates varies by no more than 100,000. Therefore, it reflects that the memory consumption results are the same within a dataset. Furthermore, it demonstrates the stability of the proposed HUSP-SP: its memory consumption performance is not as sensitive to changes in $\xi$. From the study result comparison between datasets SynDataset80K and SynDataset160K, we can also find that HUSP-SP has good scalability in terms of database size. In conclusion, the proposed upper bound TRSU significantly contributes to improving the efficiency of the algorithms.

\section{Conclusion}  \label{sec:conclusion}

Due to its more comprehensive consideration of the sequence data, database-based sequence mining has played an important role in the domain of knowledge discovery in databases. In general, utility mining takes frequency, sequential order, and utility into consideration, while the combinatorial explosion of sequences and utility computation make utility mining a NP-hard problem. This article proposes a novel HUSP-SP algorithm that addresses the problem more efficiently than the existing methods. HUSP-SP developed the compact seq-arrays to store the necessary information from sequence data. Besides, the projected database structure, namely seqPro was designed to efficiently calculate the utilities and upper bound values of candidate patterns. Furthermore, a new tight utility upper bound, namely TRSU, and two search space pruning strategies are proposed to improve the mining performance of HUSP-SP. Extensive experimental results on both synthetic and real-life datasets show that the HUSP-SP algorithm outperforms the state-of-the-art algorithms, e.g., HUSP-ULL. In the future, an interesting direction is to redesign HUSP-SP and develop a parallel and distributed version, for example, utilizing MapReduce or Spark to discover the interesting HUSPs on large-scale databases in distributed environments.

\section*{Acknowledgment}

This research was supported in part by the National Natural Science Foundation of China (Nos. 62002136 and 62272196), Natural Science Foundation of Guangdong Province (Nos. 2020A1515010970 and 2022A1515011861), Shenzhen Research Council (No. GJHZ20180928155209705), and NSF Grants (Nos. III-1763325, III-1909323, and SaTC-1930941).

\bibliographystyle{ACM-Reference-Format}
\bibliography{HUSPSP.bib}
\end{document}